\documentclass[11pt]{article}

\usepackage[LY1]{fontenc}

\usepackage{palatino} 
\usepackage{bbm} 
\usepackage{tablefootnote} 
\usepackage{physics}
\usepackage{float}
\usepackage{amssymb} 
\usepackage{amsthm}
\usepackage[hidelinks]{hyperref} 
\usepackage{cleveref} 
\usepackage{xspace} 
\usepackage[sort&compress,sectionbib,numbers]{natbib}
\usepackage[numbib]{tocbibind}
\usepackage{ulem} 
\DeclareMathOperator\arctanh{arctanh}

\hypersetup{pdfpagemode=UseNone}

\newtheorem{theorem}{Theorem}[section]
\newtheorem{lemma}[theorem]{Lemma}

\theoremstyle{definition}
\newtheorem{definition}[theorem]{Definition}

\newtheorem{remark}[theorem]{Remark}
\newtheorem{fact}[theorem]{Fact}

\newcommand{\F}{\mathrm{F}_\mathrm{U}} 
\newcommand{\FCL}{\mathrm{F}_{\mathrm{cl}}} 
\newcommand{\FH}{\mathrm{F}_{\mathrm{H}}} 
 
\newcommand{\FGM}{\mathrm{F}_{\mathrm{M}}} 
\newcommand{\newfidelity}{Matsumoto fidelity\xspace} 
\newcommand{\standard}{Uhlmann\xspace} 
\newcommand{\newfidelities}{Matsumoto fidelities\xspace}

\newcommand{\Dpishort}{Monotonicity} 

\newcommand{\R}{\mathbb{R}}
\renewcommand{\tr}{\Tr}
\renewcommand{\succcurlyeq}{\succeq}
\renewcommand{\preccurlyeq}{\preceq}

\newcommand{\Id}{\mathbbm{1}}

\newcommand{\inner}[2]{\langle #1, #2 \rangle}

\renewcommand{\Tr}{\mathrm{Tr}} 
\newcommand{\half}{1/2} 
\newcommand{\comment}[1]{{}} 
\newcommand{\eps}{\varepsilon} 
\newcommand{\Range}{\mathrm{Image}}

\DeclareMathOperator{\diag}{diag} 

\numberwithin{equation}{section}

\usepackage{color,graphicx}

\title{A fidelity measure for quantum states based on the matrix geometric mean} 

\begin{document}

\author{
	Sam Cree\thanks{Stanford University, Stanford, California, USA  
and 
Perimeter Institute for Theoretical Physics, Waterloo, Ontario, Canada, \url{scree@stanford.edu}}
\and
Jamie Sikora\thanks{
Virginia Polytechnic Institute and State University, Blacksburg, Virginia, USA 
and 
Perimeter Institute for Theoretical Physics, Waterloo, Ontario, Canada, \url{sikora@vt.edu}}  
}  

\date{{January 12, 2021}} 

\maketitle 

\vspace{-20pt}
  
\begin{abstract} 
	Uhlmann's fidelity function is one of the most widely used similarity measures in quantum theory. 
	One definition of this function is that it is the minimum classical fidelity associated with a quantum-to-classical measurement procedure of two quantum states. 
	In 2010, Matsumoto introduced another fidelity function which is dual to Uhlmann's in the sense that it is the maximimum classical fidelity associated with a classical-to-quantum preparation procedure for two quantum states.
	Matsumoto's fidelity can also be defined using the well-established notion of the matrix geometric mean. 
 In this work, we examine Matsumoto's fidelity through the lens of semidefinite programming to give simple proofs that it possesses many desirable properties for a similarity measure, including monotonicity under quantum channels, joint concavity, and unitary invariance. 
Finally, we provide a geometric interpretation of this fidelity in terms of the Riemannian space of positive definite matrices, and show how this picture can be useful in understanding some of its peculiar properties.    
\end{abstract} 


\section{Introduction} 
  
In many applications of quantum information, it is important to be able to demonstrate that two quantum states are ``close'' to one another in some sense.
For example, one may wish to demonstrate that experimental data or numerical simulations closely approximate those from another given state, or to verify the validity of a quantum algorithm or error-correction procedure. 
Thus it is useful to find comparison measures $f(\rho ,\sigma )$ that represent the similarity or distance between two quantum states with density matrices $\rho  $ and $\sigma $.

A common way to develop comparison measures for quantum states is to start with a comparison measure of classical probability distributions and look for a quantum counterpart.
A state described by a density matrix $\rho $ with eigenvalues $\{p_i\}$ can be obtained by preparing each of its eigenstates $\ket{\psi _i}$ with probability equal to the corresponding eigenvalue $p_i$. 
In this way, $\left\{ p_i \right\}$ is a probability distribution associated with $\rho $. 
Any quantum state $\sigma $ that commutes with $\rho $ shares a set of eigenstates, meaning that both states can be simultaneously interpreted as classical probabilistic mixtures over those eigenstates.
It is reasonable to define a function of two density matrices $f(\rho ,\sigma )$ to be a \emph{quantization} of a classical comparison measure of two probability distributions, $f_{cl}(\left\{ p_i \right\},\left\{ q_i \right\})$, if it agrees in the case of commuting states, i.e., 
\begin{align}
	f\text{ \emph{quantizes} }f_{cl} \;\; \text{ if } \;\; [\rho ,\sigma ]=0 \implies f(\rho ,\sigma ) = f_{cl} ( \left\{ p_i \right\}, \left\{ q_i \right\} ),
	\label{eq:qtz}
\end{align}
with $\left\{ p_i \right\}$ and $\left\{ q_i \right\}$ the eigenvalues of $\rho $ and $\sigma $ respectively (labelled according to some indexing of a shared eigenbasis).

Many well-known comparison measures from quantum information theory arise as quantizations of classical quantities.
For example, the trace distance is the quantization of a classical quantity known as the \emph{Kolmogorov distance} between two probability distributions, the quantum relative entropy is a quantization of the classical \emph{Kullback-Leibler divergence}, and the quantum fidelity quantizes the 
\textit{classical fidelity} (or \textit{Bhattacharyya coefficient}). 
The trace distance, quantum relative entropy, and quantum fidelity have all found widespread application within the field of quantum information, with each being particularly convenient for calculations in specific contexts. 
This motivates the study of alternative comparison measures for quantum states, to develop a wider range of available tools for applications of quantum information theory.

In this endeavour, one can exploit a generic feature of quantization, namely its non-uniqueness\footnote{
	An analogous example of non-uniqueness being useful is in the quasi-probability distribution formulation of continuous-variable quantum mechanics, in which each quantum state is represented by a probability-like distribution over phase space (analogous to a classical probability distribution).
Depending on whether quantum operators are ordered normally, anti-normally, or symmetrically when defining this distribution, one obtains either the Glauber-Sudarshan $\mathrm P$ representation, the Husimi $\mathrm \sigma $ representation, or the Wigner quasi-probability distribution respectively.
Each of these distributions are useful in different contexts as they represent qualitatively distinct information about the quantum state -- for example, the Glauber-Sudarshan $\mathrm P$ representation is the best indicator of non-classicality, the Wigner distribution leads to the simplest evaluation of expectation values, and the Husimi $\mathrm \sigma $ representation is the only strictly non-negative distribution of the three.
};
for a given classical quantity, there is generally an infinite family of quantizations that agree in the classical limit.
For example, consider the classical fidelity, defined for probability distributions $\left\{ p_i \right\}$ and $\left\{ q_i \right\}$ as 
\begin{equation} 
	\FCL(\left\{ p_i \right\},\left\{ q_i \right\}) := \sum_{i=1}^n \sqrt{p_i q_i}. \label{eq:fclh}
\end{equation}  
By rewriting this in a more symmetric way, and then replacing the probability distributions with density matrices (and the sum with a trace), one obtains the (standard) {quantum} fidelity due to Uhlmann~\cite{uhlmann} 
\begin{align}
	\FCL(\left\{ p_i \right\},\left\{ q_i \right\}) = \sum_{i=1}^n \sqrt{\sqrt{p_i} q_i\sqrt{p_i}} \quad \longrightarrow\quad  \F(\rho ,\sigma ) := \tr \left( ( \rho ^{1/2} \sigma  \rho ^{1/2})^{1/2} \right)= \| \rho ^{1/2} \sigma ^{1/2} \|_1 , \label{eq:fu}
\end{align} 
where $\| \cdot \|_1$ is the trace norm, defined as\footnote{We use the notation $M^\dagger$ for the adjoint (or conjugate transpose) of $M$. For positive semidefinite matrix $A$, $A^{1/2}$ is the unique positive semidefinite matrix $M$ such that $M^2 = A$.} $\| M \|_1 :=\tr((MM^{\dagger})^{1/2})$. 
We refer to this from here on as the \emph{\standard fidelity} to distinguish it from other fidelity measures in this work.
This well-known measure of similarity for quantum states has many physically desirable and mathematically convenient properties, including those below (which we discuss in more detail later)  
\begin{itemize}
  \setlength\itemsep{0em}
	\item Symmetry in its inputs, 
	\item Ranges from 0 to 1,
	\item Attains 1 if and only if the states are identical, 
	\item Attains 0 for states that are orthogonal i.e.\ $\tr(\rho  \sigma  )=0$,
	\item Monotonicity under quantum channels, $\F(\mathcal{E}(\rho ),\mathcal{E}(\sigma )) \geq \F(\rho ,\sigma )$,
	\item Unitary invariance,
	\item Joint concavity in its inputs,
	\item Additivity over direct sums, and
	\item Multiplicativity over tensor products.
\end{itemize} 
One can verify that for commuting states, $\F$ is a valid quantization of $\FCL$ in the sense of \Cref{eq:qtz}.
For the purposes of this work, we say that any similar quantity that quantizes the classical fidelity is a ``quantum fidelity''\footnote{A famous paper of Jozsa \cite{josza} lists an alternative set of desirable axioms for a reasonable fidelity measure, one of which is that it should equal $\bra{\psi }\rho \ket{\psi }$ for a pure state $\ketbra{\psi }$ and a mixed state $\rho $.
	By choosing to relax this axiom in favour of fixing the classical limit as $\FCL$, we are simply studying a different family of quantities to those fitting into Jozsa's framework. 
	See \cite{fidelities} for a thorough discussion of a large range of fidelities in terms of these axioms.
We discuss how the fidelities discussed in this work behave when one state is pure in Subsection~\ref{sect:onepure}.
}.
Note that many quantum fidelities are not particularly interesting or useful; for example, a family of quantum fidelities is given by $\F(\rho ,\sigma ) + f([\rho ,\sigma ])$, where $f$ is an arbitrary function of the commutator satisfying $f(0)=0$.
But this family of fidelities generally fails many desirable basic properties such as those listed above. 

In fact, we only know of three quantum fidelities discussed in the quantum literature that satisfy the list of properties above, the first being the Uhlmann fidelity. The second is the Holevo fidelity\footnote{This was first discovered by Holevo \cite{holevo} and actually predates the \standard fidelity. 
	It has been studied sporadically under a number of different names: just-as-good fidelity \cite{recoverability}, pretty-good fidelity \cite{prettygood}, $A$-fidelity \cite{a}, overlap information \cite{bhip}, and affinity \cite{affinity}, and is also directly related to the quantum Tsallis relative entropy \cite{tsallis3}, and the R\'enyi relative entropy \cite{comparison}.
See \cite{recoverability} for more discussion of the history of this quantity.}, defined as

\begin{equation}  
	\FH(\rho ,\sigma ) := \tr( \rho ^{1/4} \sigma ^{1/2} \rho ^{1/4}),
	\label{eq:fh}
\end{equation}  
which is distinct from $\F$ (which is suggested by the fact that the matrix square root does not distribute over multiplication; in general 
$(ABA)^{1/2} \neq A^{1/2} B^{1/2} A^{1/2}$).

The main focus of this work is to study the third known fidelity satisfying the list of above properties.
We call it the \textit{\newfidelity} after its introduction by Matsumoto in \cite{matsumoto}, which is defined as 
\begin{equation}
	\FGM(\rho ,\sigma ) := \tr(\rho  \#\sigma  ) 
\;\; \text{ where } \;\; \rho \# \sigma  := \rho ^{1/2} (\rho  ^{-1/2} \sigma  \rho ^{-1/2})^{1/2} \rho ^{1/2},	\label{eq:res}
\end{equation}
for invertible quantum states $\rho $ and $\sigma $. 
If $\rho $ or $\sigma $ is singular, $\rho  \# \sigma $ can be defined via the limit
\begin{align}
\rho  \# \sigma  := \lim_{\varepsilon \to 0} (\rho +\varepsilon \Id ) \# (\sigma +\varepsilon \Id). \label{eq:limit}
\end{align}
This quantity satisfies all of the properties listed above.  
The binary operation $ \# $ is known as the \textit{matrix geometric mean} \cite{geometricmean} (see also \cite{bhatia}) as it naturally extends the notion of geometric mean for two positive numbers to the case of positive definite matrices\footnote{$A$ is positive definite (denoted $A \succ 0$) if $\vb v^{\dagger }\! A \vb v>0$ for all nonzero vectors $\vb v$. Similarly, $A$ is positive \textit{semidefinite} (denoted $A\succcurlyeq 0$) if $\vb v^{\dagger }\! A \vb v\geq 0$ for all vectors $\vb v$.}.
The matrix geometric mean has intricate connections to the geometry of quantum state space and also to quantum information theory (see \cite{tsallis1,uhlmann2010,tsallis2,sdpquantum7,optimality} for examples), and even plays some hidden roles in the study of the \standard fidelity (which we discuss in the next section). 
   
\paragraph{Authors' note:} We were made aware\footnote{We thank Mark M. Wilde for pointing this work out to us.} of the works by Matsumoto \cite{matsumoto,matsumotoSDP,matsumotof} in which he introduces this quantity while in the final stages of preparing this work. 
Therefore, although this work was done independently, it may be viewed as a review (instead of introduction) of this fidelity function.  
However, this work does offer novel contributions and perspectives which we summarize at the end of the introduction. 
  
\paragraph{A tale of two fidelities.}
Matsumoto showed several interesting ways in which the \newfidelity acts as a ``dual'' to the \standard fidelity.
The first relates to operational interpretations of the two quantities, which we now discuss.  
We say that two probability distributions $\left\{ p_i \right\} $ and $ \left\{ q_i \right\}$ can be \textit{obtained by a measurement} of states $\rho $ and $\sigma $ if there exists a POVM $\left\{M_i \right\}$ such that $\tr(M_i \rho ) = p_i$ and $\tr(M_i \sigma ) = q_i$.
The \standard fidelity has an operational interpretation as the \emph{minimal} classical fidelity between two probability distributions $\left\{p_i\right\}$ and $\left\{ q_i \right\}$ that can be \emph{obtained by a measurement} of $\rho $ and $\sigma $.

Consider now the procedure of preparing a quantum state from an initial classical probability distribution, which is in some sense the ``reverse'' of obtaining classical probabilities via measurement of a quantum state. 
We say that two probability distributions $\left\{ p_i \right\} $ and $ \left\{ q_i \right\}$ can \textit{prepare} states $\rho $ and $\sigma $ if there exists a quantum channel $\mathcal{E}$ such that\footnote{Here the $\diag$ operator simply puts a vector on the diagonal of a diagonal matrix.} $\mathcal{E}(\diag\left\{ p_i \right\}) = \rho $ and $\mathcal{E}(\diag\left\{ q_i \right\})=\sigma $.
Then the \newfidelity is the \emph{maximal} possible classical fidelity between classical probability distributions that can \emph{prepare} quantum states $\rho $ and $\sigma $. 
In other words, the \newfidelity is the maximum classical fidelity for a classical-to-quantum preparation procedure, whereas the \standard fidelity is the minimum classical fidelity for a quantum-to-classical measurement procedure.

Another interesting feature of these two fidelities is that they completely bound the space of monotonic fidelities.
{The \newfidelity is the smallest possible quantization of the classical fidelity satisfying monotonicity under quantum channels and the Uhlmann fidelity is the largest, i.e.\ if $\mathrm{F}$ is a quantization of $\FCL$ satisfying monotonicity, then for any $\rho $ and $\sigma $, we have}
\begin{equation} 
\FGM(\rho ,\sigma ) \leq \mathrm{F}(\rho ,\sigma ) \leq \F(\rho, \sigma). 
\end{equation}   
See \cite{matsumoto} for details.  
Thus although these two quantities agree for commuting states, they contrast dramatically, and in a sense, \emph{maximally}, for states that fail to commute. 

An example of this dramatic difference is for non-commuting pure states.
One can show that the \newfidelity is \textit{exactly zero} for any two distinct pure states, which can be verified from \Cref{eq:limit}.
This is in contrast to the \standard and Holevo fidelities, as we have for pure states $\ketbra{\psi }$ and $\ketbra{\phi }$,
\begin{align}
	\F(\ketbra{\psi} ,\ketbra{\phi }) &=  \qty|\braket{\psi }{\phi }|\label{eq:pure} \\
\quad \FH(\ketbra{\psi} ,\ketbra{\phi }) &= \qty|\braket{\psi }{\phi }|^2
\label{eq:pureh} .
\end{align}
We can rewrite the \newfidelity of two pure states in the following suggestive way to draw parallel with these\footnote{This may suggest the existence of some family of fidelities $\mathrm{F}_p$ such that $\mathrm{F}_p(\ketbra{\psi },\ketbra{\phi }) = \qty|\braket{\psi }{\phi }|^p$, with the \newfidelity ($p=\infty $), the Holevo fidelity ($p=2$), and the \standard fidelity ($p=1$) as special cases.},
\begin{align}
	\quad \FGM(\ketbra{\psi} ,\ketbra{\phi }) &=  \qty|\braket{\psi }{\phi }|^\infty .
	\label{eq:puregm} 
\end{align} 
The case of pure states, in which the three fidelities starkly disagree, can be thought of as the opposite to the classical limit in which they all coincide.
Pure states are in a sense the ``most quantum'' states, being fully coherent and not relying on decoherence or classical probabilistic mixtures to be prepared. 
Thus the differences between the three quantities  function as an indirect probe into the states' non-commutativity, or of their ``relative quantumness''. 

\label{sect:fg} 
      
\paragraph{Semidefinite programming.}
\label{sect:sdp} 
   
Semidefinite programming is a well-behaved class of optimization problems which have seen countless applications in the study of quantum theory, including 
convex geometry~\cite{sdpquantum7,watrous2009semidefinite,watrous2013simpler,sdpquantum8,sdpquantum9}, 
thermodynamics~\cite{jamie1}, 
computational complexity theory~\cite{jamie2,sdpquantum1,sdpquantum2,sdpquantum3}, 
cryptography~\cite{sdpquantum11,jamie3,jamie5,sdpquantum10}, 
Bell non-locality~\cite{jamie4,sdpquantum4,sdpquantum5}, 
and entanglement~\cite{sdpquantum6}, to just name a few. 
Fortunately, the Matsumoto fidelity can be formulated as a semidefinite program (abbreviated as SDP) which allows a convenient prescription for its calculation, and also provides a useful analytical definition with which many of its properties can be easily proven.  
  
The \newfidelity can be formulated as the optimal objective function value of the following SDP~\cite{matsumotoSDP} 
\begin{equation} \label{eq:FGMSDP}
	\FGM (\rho ,\sigma ) = \sup \left\{ \tr(W) : 
\left[ \begin{array}{cc} 
\rho & W \\ 
W & \sigma  
\end{array} \right] \succcurlyeq 0  
\right\}, 
\end{equation}
which conveniently avoids limits for non-invertible quantum states.
This form is especially useful when  one wishes to optimize $\FGM(\rho, \sigma)$ when one or both of the input states are not fixed.  
Furthermore, this definition is simple compared with \Cref{eq:limit}, which as we demonstrate shortly, makes it an easier starting point to prove many of its properties (especially since we can avoid dealing with limits for the most part). 

It is worth noting that it bears striking similarities to an SDP for the \standard fidelity%
\footnote{Note that there is an SDP formulation for the Holevo fidelity as well (see for example \cite{sdpquantum7}) but it has a slightly more complicated structure and we do not study it in this work.}~\cite{watrous2013simpler}
\begin{align} \label{eq:FSDP}
	\F(\rho ,\sigma ) 
& = \sup \left\{ \frac{1}{2} \tr(X) + \frac{1}{2} \tr(X^{\dagger }) :   
\left[ \begin{array}{cc} 
\rho & X \\ 
X^{\dagger } & \sigma  
\end{array} \right] \succcurlyeq 0 \right\} .
\end{align}
We note that the only difference is that the variable $X$ in the formulation above need not be Hermitian. 
This also immediately implies that the Uhlmann fidelity is no less than the \newfidelity for all choices of inputs (which we formally prove later). 
One might notice that this also follows from many other characterizations of these two quantities discussed throughout this work.  
  
\paragraph{Contributions of this work.} 

Although many of the properties presented here can be found throughout the works \cite{matsumoto}, \cite{matsumotof}, and \cite{matsumotoSDP}, we approach their proofs from a very different perspective.  
In particular, we use semidefinite programming to bypass the subtleties that are otherwise required for non-invertible density matrices, allowing for straightforward derivations of many of its properties. 
Also, we provide a novel geometric interpretation in terms of the Riemannian metric on the space of positive definite matrices.
Finally, we hope that this introduction helps publicize some of the particularly interesting results in \cite{matsumoto,matsumotoSDP,matsumotof}, as we were surprised to see that they have not received more attention in the literature\footnote{For example, \cite{fidelities} provides an otherwise comprehensive review of different quantum fidelities and their properties, but the Matsumoto fidelity does not make an appearance.}. 
We refer the interested reader to  \cite{matsumoto,matsumotoSDP,matsumotof} for further reading on the topic. 
These works also introduce a generalization to a family of quantum $f$-divergences, which has been further studied, see e.g.\ \cite{divergences,geometricrenyi}\footnote{{One other interesting note is that each of the three fidelities discussed here arise as special cases of three well-known families of quantum R\'enyi relative entropy. The Uhlmann fidelity is a special case of the sandwiched R\'enyi relative entropy with $\alpha =\frac{1}{2}$ \cite{Renyi2,Renyi1}; the Holevo fidelity a special case of the Petz-R\'enyi relative entropy with $\alpha =\frac{1}{2}$ \cite{petz2}; and the Matsumoto fidelity is a special case of the geometric R\'enyi divergence with $\alpha =\frac{1}{2}$ \cite{matsumotof}.}}.
    
\paragraph{Results and organization.} 
This work is structured as follows.
In \Cref{sect:gm}, we give some background information on the matrix geometric mean and discuss some of its hidden appearances in the study of the Uhlmann fidelity.  
In \Cref{sect:properties}, we prove a number of properties of the \newfidelity and compare them with those of the \standard fidelity and the Holevo fidelity. 
In particular, in Subsection~\ref{sect:cases}, we provide some concrete examples of how the three fidelities differ for certain special cases, namely when one state is pure or when both states are qubits. 
In \Cref{sect:geometry}, we introduce the geometric intepretation of the \newfidelity, and discuss it in detail for the qubit case.
We conclude in \Cref{sect:conclusion} and discuss open questions for future work.  


\section{Background on the matrix geometric mean} \label{sect:gm} 
   
In this section, we discuss properties of the matrix geometric mean and present a few places where it shows up in the study of the \standard fidelity.         


\subsection{Definitions and properties of the matrix geometric mean}  \label{sect:mgm}
    
We begin by recalling the definition from the introduction.
\begin{definition}[Matrix geometric mean] 
	For two positive \emph{definite} matrices $A$ and $B$, their matrix geometric mean is given by the formula
\begin{equation} 
A \# B := A^{\half} \left( A^{-\half} B A^{- \half} \right)^{\half} A^{\half}. 
\end{equation} 
For two positive \emph{semidefinite} matrices $A$ and $B$, we define their matrix geometric mean as 
\begin{equation} 
A \# B = \lim_{\eps \to 0} (A_{\eps} \# B_{\eps}) 
\label{eq:limit2}
\end{equation} 
where we use the notation
\begin{align} \label{shorthand}
X_{\eps} := X + \eps \Id  
\end{align}
here and throughout the paper for brevity. 
It can be shown that this limit does exist and thus {\Cref{eq:limit2}} is well-defined.  
\end{definition} 
   
The matrix geometric mean has a number of nice properties, some of which we present below.  

\begin{fact}[Properties of the matrix geometric mean] 
\label{basicprops} 
For any positive \emph{definite} matrices $A$ and $B$, we have the following properties:  
\begin{enumerate} 
\item Symmetry: $A \# B = B \# A$.  
\item Hermitian and positive definite: $A\# B\succ 0$.
\item $A \# B = A^{\half} B^{\half}$ if $A$ and $B$ commute. 
\item If $X$ is invertible, then $X(A \# B)X^\dagger  = (XAX^\dagger ) \# (XBX^\dagger )$.  
\item Inverses: $(A \# B)^{-1} = A^{-1} \# B^{-1}$.  
\item If $A, B, C, D$ satisfy $A \succcurlyeq B \succ 0$ and $C \succcurlyeq D \succ 0$, then $A \# C \succcurlyeq B \# D$.  
\item Arithmetic-geometric mean inequality: $\frac{1}{2}(A+B) \succcurlyeq A\# B$. 
\item {Distributive property:} {$(A \otimes B) \# (C \otimes D) = (A \# C) \otimes (B \# D) $.} 
\item For any positive (not necessarily completely positive) map $\Phi$, we have 
\begin{equation} 
\Phi(A) \# \Phi(B) \succcurlyeq \Phi(A \# B). 
\end{equation}  
\end{enumerate} 
\end{fact} 
Properties 1--7 can be found in \cite{karcher}, Property 8 can be verified directly, and Property {9} was proven in \cite{Ando}. 
The interested reader is referred to the book \cite{bhatia} for a nice discussion on the topic.


\subsection{The matrix geometric mean and the \standard fidelity}  

Although we introduce the matrix geometric mean in order to study the \newfidelity, it also plays a role in the study of other quantum information quantities, such as the \standard fidelity.
The \standard fidelity can be expressed \cite{uhlmann2010} as 
\begin{align}
	\F(\rho ,\sigma ) = \tr\left(\rho \left( \rho ^{-1} \# \sigma \right)\right) 
\end{align}
when $\rho $ and $\sigma $ are invertible. 
This suggests the following fact about the gradient of the \standard fidelity (which is proven in \cite{optimality}). 
\begin{fact} 
For a fixed invertible quantum state $\sigma $, if we define  
\begin{equation} 
	g(\rho ) := \F(\rho ,\sigma ) \text{ for } \rho \text{ invertible},  
\end{equation} 
then we have 
\begin{equation} 
	\nabla g(\rho ) = \frac{1}{2} \left( \rho ^{-1} \# \sigma  \right).  
\end{equation}
\end{fact} 

The matrix geometric mean also appears in the characterization of the \standard fidelity given by Alberti \cite{alberti}, below. 
\begin{fact}[Alberti's Theorem] 
For any quantum states $\rho $ and $\sigma $, we have 
\begin{equation} 
	\F(\rho ,\sigma )^2 = \inf_{\tau \succ 0} \inner{\tau}{\rho } \inner{\tau^{-1}}{\sigma }. 
\end{equation}
\end{fact} 

It is easy to verify that if $\rho $ and $\sigma $ are invertible, then an optimal choice of $\tau$ is 
\begin{equation} 
\tau := \rho ^{-1} \# \sigma . 
\end{equation} 

Finally, as mentioned in the introduction, we have that the fidelity is equal to the so-called \emph{measurement fidelity}, described below.  
\begin{fact} 
For any quantum states $\rho $ and $\sigma $ and POVM $(M_1, \ldots, M_m)$, we have that 
\begin{equation} \label{msmtfid}
	\F(\rho ,\sigma ) \leq \FCL(p, q) 
\end{equation} 
where $p_i := \inner{M_i}{\rho }$ and $q_i := \inner{M_i}{\sigma }$ for all $i \in \{ 1, \ldots, m \}$. 
Moreover, there exists a POVM such that the above holds with equality. 
\end{fact}  
   
It turns out that the POVM that makes \eqref{msmtfid} hold with equality is the measurement in the basis of $\rho ^{-1} \# \sigma $, when $\rho $ and $\sigma $ are invertible. 
      
       
\section{Properties of the \newfidelity} \label{sect:properties} 
   
Due to its prevalence in quantum information, many useful properties of the \standard fidelity have been discovered.
In this section we list some properties of the \newfidelity and organize the properties with respect to how they compare with the \standard fidelity. 
To keep the presentation clean, we only compare it to the \standard fidelity. 
In Table~\ref{tab:props}, we summarize the properties and then also compare them to those of the Holevo fidelity.  

\begin{remark} 
Note that many of the properties presented in this section hold for general positive semidefinite matrices, i.e., they do not require the unit trace condition. 
However, we present and prove them for quantum states. 
It should be clear to the interested reader to see which require the unit trace condition and which do not. 
\end{remark}
   
We begin by proving the validity of the semidefinite program for the \newfidelity presented in \Cref{eq:FGMSDP}, justifying its use throughout this section. 
It follows almost immediately from the given lemma. 

\begin{lemma} \label{handylemma}
Given positive semidefinite matrices $P$ and $Q$, if $W$ satisfies 
\begin{equation} \label{eqlem32}
\left[ \begin{array}{cc} 
P & W \\ 
W & Q  
\end{array} \right] \succcurlyeq 0 
\end{equation} 
then $W \preceq P \# Q$. 
Moreover, $P \# Q$ satisfies \Cref{eqlem32}. 
\end{lemma} 

We prove this in the appendix, restated as \Cref{handylemmaproof}.

\begin{lemma}[SDP formulation]  \label{SDPlemma} 
For all quantum states $\rho$ and $\sigma$, we have 
\begin{equation} 
	\FGM (\rho ,\sigma ) = \max \left\{ \Tr(W) : 
\left[ \begin{array}{cc} 
\rho & W \\ 
W & \sigma  
\end{array} \right] \succcurlyeq 0   
\right\}. \label{eq:FGSDP} 
\end{equation} 
Moreover, the maximum is attained at  
$W = \rho \# \sigma $.   
\end{lemma} 

\begin{proof}  
 
Let $\alpha$ be the optimal value of the above SDP. 
If $W$ is feasible, by \Cref{handylemma} we know that $W \preceq \rho \# \sigma$. 
Thus, 
\begin{equation} 
\alpha \leq \Tr(\rho \# \sigma) = \FGM(\rho, \sigma). 
\end{equation} 
On the other hand, $\rho \# \sigma$ is feasible, and thus 
\begin{equation} 
\alpha \geq \Tr(\rho \# \sigma) = \FGM(\rho, \sigma)  
\end{equation} 
as desired. 
 
\end{proof} 
      
We use the above SDP formulation to show many of the following properties of the \newfidelity.  

\begin{lemma} \label{fidbound}
	For any quantum states $\rho $ and $\sigma $, we have $\FGM(\rho ,\sigma ) \leq \F(\rho ,\sigma )$.  
\end{lemma}  

\begin{proof} 
It was shown in \cite{watrous2013simpler} that 
\begin{equation} 
	\F(\rho ,\sigma ) = \max \left\{ \frac{1}{2} \Tr(X) + \frac{1}{2} \Tr(X^\dagger ) : 
\left[ \begin{array}{cc} 
\rho & X \\ 
X^\dagger  & \sigma  
\end{array} \right] \succcurlyeq 0   
\right\}. \label{eq:FSDP} 
\end{equation} 
By adding the constraint $X = X^\dagger $ to the above SDP, we recover SDP~(\ref{eq:FGSDP}) which exactly characterizes the \newfidelity.  
Since the above SDP is a maximization and is less constrained than SDP~(\ref{eq:FGSDP}), the optimal objective function value can only increase. 
\end{proof}  

\begin{remark} 
Note that many of the properties discussed shortly in this work can be proven \emph{for invertible quantum states} by invoking known properties of the matrix geometric mean. 
However, we prove them here for \emph{all quantum states}, and many of our proofs are simple and do not rely on known properties (although we do occasionally use them). 
For instance, by dealing with limits in the proof of Lemma~\ref{handylemma} we avoid dealing with limits in many of the upcoming proofs.
\end{remark}


\subsection{Properties shared with the \standard fidelity} 

Here we present some of the properties of the \newfidelity that are shared with the \standard fidelity. 

\begin{lemma}[Symmetry] 
For all quantum states $\rho $ and $\sigma $, we have 
\begin{equation} 
	\FGM(\rho ,\sigma ) = \FGM(\sigma , \rho). \label{item:sym} 
\end{equation} 
\end{lemma}
\begin{proof} 
This follows immediately from the SDP formulation \eqref{eq:FGSDP}. 
\end{proof}
  
\begin{lemma}[Bounds] \label{lem:bounds}
For all quantum states $\rho $ and $\sigma $, we have 
\begin{equation}
	0 \leq \FGM(\rho ,\sigma ) \leq 1. \label{item:bounds} 
\end{equation} 
\end{lemma} 

\begin{proof}
Firstly, since $W = 0$ is always a feasible solution to the SDP \eqref{eq:FGSDP}, we have that ${\FGM(\rho, \sigma) \geq 0}$ for all states $\rho$ and $\sigma$. 
Since $\FGM(\rho ,\sigma ) \leq \F(\rho ,\sigma )$ and we have $\F(\rho ,\sigma ) \leq 1$, for all quantum states, the result holds. 
\end{proof}

\begin{lemma}[Unity condition] 
For all quantum states $\rho $ and $\sigma $, we have 
\begin{equation}
	\FGM(\rho ,\sigma ) = 1 
\, \text{ if and only if } \, 
\rho  = \sigma . \label{item:one} 
\end{equation} 
\end{lemma}

\begin{proof} 
	If $\FGM(\rho ,\sigma ) = 1$, then we have $\F(\rho ,\sigma ) = 1$ which implies $\rho  = \sigma $. 
	Conversely, if $\rho  = \sigma $, then $W = \rho $ is a feasible solution to the SDP 
	
\begin{equation} 
	\FGM (\rho, \rho) = \max \left\{ \Tr(W) : 
\left[ \begin{array}{cc} 
\rho & W \\ 
W & \rho  
\end{array} \right] \succcurlyeq 0   
\right\} 
\end{equation} 
	certifying $\FGM(\rho , \rho ) \geq 1$.
The result now holds by Lemma~\ref{lem:bounds}.  
\end{proof}

\begin{lemma}[Additivity] \label{lemm:add}
	For all quantum states $\rho_1$, $\rho_2$, $\sigma_1$, $\sigma_2$ (of compatible dimensionalities), and scalars $\lambda_1, \lambda_2 \in (0,1)$ satisfying $\lambda_1 + \lambda_2 = 1$, we have 
\begin{equation} 
	\FGM( \lambda_1 \rho_1 \oplus  \lambda_2 \rho _2, \lambda_1 \sigma_1 \oplus \lambda_2 \sigma  _2) = \lambda_1 \FGM(\rho _1, \sigma _1) +  \lambda_2 \FGM(\rho  _2, \sigma _2). 
\end{equation} 
\end{lemma}  

\begin{proof}  
Let $W_1$ be an optimal solution to the SDP 
\begin{equation} 
	\FGM (\rho_1, \sigma_1) = \max \left\{ \Tr(W) : 
\left[ \begin{array}{cc} 
\rho_1 & W \\ 
W & \sigma_1  
\end{array} \right] \succcurlyeq 0   
\right\} \label{eq:FGSDP11} 
\end{equation} 
and let $W_2$ be an optimal solution to the SDP 
\begin{equation} 
	\FGM (\rho_2, \sigma_2) = \max \left\{ \Tr(W) : 
\left[ \begin{array}{cc} 
\rho_2 & W \\ 
W & \sigma_2  
\end{array} \right] \succcurlyeq 0   
\right\}. \label{eq:FGSDP12} 
\end{equation} 
It is straightforward to see that $\lambda_1 W_1 \oplus \lambda_2 W_2$ is feasible for the SDP 
\begin{equation} 
	\FGM (\lambda_1 \rho_1 \oplus \lambda_2 \rho_2, \lambda_1 \sigma_1 \oplus \lambda_2 \sigma_2) = \max \left\{ \Tr(W) : 
\left[ \begin{array}{cc} 
\lambda_1 \rho_1 \oplus \lambda_2 \rho_2 & W \\ 
W & \lambda_1 \sigma_1 \oplus \lambda_2 \sigma_2 
\end{array} \right] \succcurlyeq 0   
\right\}. \label{eq:FGSDP13} 
\end{equation} 
Thus, 
\begin{equation} 
\FGM (\lambda_1 \rho_1 \oplus \lambda_2 \rho_2, \lambda_1 \sigma_1 \oplus \lambda_2 \sigma_2) \geq \tr(\lambda_1 W_1 \oplus \lambda_2 W_2) = \lambda_1 \FGM(\rho _1, \sigma _1) +  \lambda_2 \FGM(\rho  _2, \sigma _2). 
\end{equation} 

Conversely, suppose 
\begin{equation}
W = \left[ \begin{array}{cc} 
W_{11} & W_{12} \\ 
W_{21} & W_{22} 
\end{array} \right] 
\end{equation} 
is an optimal solution to the SDP~(\ref{eq:FGSDP13}) (with the partitioning being clear from the context below). 
Since $W$ satisfies 
\begin{equation} 
\left[ \begin{array}{cc} 
\lambda_1 \rho_1 \oplus \lambda_2 \rho_2 & W \\ 
W & \lambda_1 \sigma_1 \oplus \lambda_2 \sigma_2  
\end{array} \right] \succcurlyeq 0,
\end{equation} 
by looking at symmetric submatrices, one can check that 
\begin{equation}
\left[ \begin{array}{cc} 
\lambda_1 \rho_1 & W_{11} \\ 
W_{11} & \lambda_1 \sigma_1
\end{array} \right] \succeq 0 
\quad \text{ and } \quad 
\left[ \begin{array}{cc} 
\lambda_2 \rho_2 & W_{22} \\ 
W_{22} & \lambda_2 \sigma_2
\end{array} \right] \succeq 0.   
\end{equation} 
Therefore, we have that $\frac{1}{\lambda_1} W_{11}$ is feasible for the SDP~(\ref{eq:FGSDP11}) and $\frac{1}{\lambda_2} W_{22}$ is feasible for the SDP~(\ref{eq:FGSDP12}). 
Thus, 
\begin{equation} 
\FGM(\rho _1, \sigma _1) \geq \frac{1}{\lambda_1} \tr(W_{11}) 
\quad \text{ and } \quad 
\FGM(\rho _2, \sigma _2)  \geq \frac{1}{\lambda_2} \tr(W_{22}) 
\end{equation} 
implying 
\begin{equation}
\FGM (\lambda_1 \rho_1 \oplus \lambda_2 \rho_2, \lambda_1 \sigma_1 \oplus \lambda_2 \sigma_2) = \tr(W_{11}) + \tr(W_{22}) 
\leq \lambda_1 \FGM(\rho _1, \sigma _1) +  \lambda_2 \FGM(\rho  _2, \sigma _2)  
\end{equation} 
as desired. 
\end{proof} 

\begin{lemma}[Multiplicativity] \label{lemm:mult}
	For all quantum states $\rho _1, \rho _2, \sigma _1, \sigma _2$ (of compatible dimensionalities), we have 
\begin{equation} 
\FGM(\rho _1 \otimes \rho _2, \sigma _1 \otimes \sigma _2) = \FGM(\rho _1, \sigma _1) \cdot \FGM(\rho  _2, \sigma _2). 
\end{equation} 
\end{lemma} 
  
Even though the matrix geometric mean behaves nicely over Kronecker products of invertible quantum states, it gets a little tricky with non-invertible states. 
For instance, in general we have 
\begin{equation} 
(\rho \otimes \sigma) + \epsilon \Id \neq (\rho + \epsilon \Id) \otimes (\sigma + \epsilon \Id). 
\end{equation}  
  
\begin{proof}[Proof of Lemma~\ref{lemm:mult}] 
Let $W_1$ be an optimal solution to the SDP 
\begin{equation} 
	\FGM (\rho_1, \sigma_1) = \max \left\{ \Tr(W) : 
\left[ \begin{array}{cc} 
\rho_1 & W \\ 
W & \sigma_1  
\end{array} \right] \succcurlyeq 0   
\right\} \label{eq:FGSDP1} 
\end{equation} 
and $W_2$ be an optimal solution to the SDP 
\begin{equation} 
	\FGM (\rho_2, \sigma_2) = \max \left\{ \Tr(W) : 
\left[ \begin{array}{cc} 
\rho_2 & W \\ 
W & \sigma_2  
\end{array} \right] \succcurlyeq 0   
\right\}. \label{eq:FGSDP2} 
\end{equation} 
We see that
\begin{equation} \label{bigmatrix}
\left[ \begin{array}{cc} 
\rho_1 & W_1 \\ 
W_1 & \sigma_1  
\end{array} \right] \otimes \left[ \begin{array}{cc} 
\rho_2 & W_2 \\ 
W_2 & \sigma_2  
\end{array} \right] \succeq 0  
\end{equation}  
since each individual matrix is positive semidefinite. 
Note that
\begin{equation} 
\left[ \begin{array}{cc} 
\rho_1 \otimes \rho_2 & W_1 \otimes W_2 \\ 
W_1 \otimes W_2 & \sigma_1 \otimes \sigma_2  
\end{array} \right] 
\end{equation} 
is a symmetric submatrix of the positive semidefinite matrix in \Cref{bigmatrix}, and is thus positive semidefinite as well. 
Therefore, $W = W_1 \otimes W_2$ 
is feasible in the SDP 
\begin{equation} 
	\FGM (\rho_1 \otimes \rho_2, \sigma_1 \otimes \sigma_2) = \max \left\{ \Tr(W) : 
\left[ \begin{array}{cc} 
\rho_1 \otimes \rho_2 & W \\ 
W & \sigma_1 \otimes \sigma_2  
\end{array} \right] \succcurlyeq 0   
\right\}. \label{eq:FGSDP3} 
\end{equation} 
Thus, 
\begin{equation} 
\FGM (\rho_1 \otimes \rho_2, \sigma_1 \otimes \sigma_2) \geq \tr(W) = \tr(W_1) \cdot \tr(W_2) = \FGM(\rho_1, \sigma_1) \cdot \FGM(\rho_2, \sigma_2). 
\end{equation} 
  
For the reverse inequality, 
we can exploit some of the previously discussed properties of the matrix geometric mean. 
For instance, for any positive semidefinite matrices $A$ and $B$, we have 
\begin{equation} \label{epsquared}
A_{\epsilon} \otimes B_{\epsilon}  \succeq (A \otimes B)_{\epsilon^2} \succ 0 
\end{equation}
recalling \Cref{shorthand} and noting the $\epsilon^2$ on the right-hand side.  
Therefore, 
\begin{align} 
\FGM(\rho_1 \otimes \rho_2, \sigma_1 \otimes \sigma_2) 
& = \tr \left( \lim_{\epsilon \to 0} (\rho_1 \otimes \rho_2)_{\epsilon^2} \# (\sigma_1 \otimes \sigma_2)_{\epsilon^2} \right) \\ 
& = \lim_{\epsilon \to 0} \tr( (\rho_1 \otimes \rho_2)_{\epsilon^2} \# (\sigma_1 \otimes \sigma_2)_{\epsilon^2} ) \\ 
& \leq \lim_{\epsilon \to 0} \tr( ((\rho_1)_{\epsilon} \otimes (\rho_2)_{\epsilon}) \# ((\sigma_1)_{\epsilon} \otimes (\sigma_2)_{\epsilon})) \;\; \text{[using 6.~from Fact~\ref{basicprops} and (\ref{epsquared})]} \\
& = \lim_{\epsilon \to 0} \tr( ((\rho_1)_{\epsilon} \# (\sigma_1)_{\epsilon}) \otimes ((\rho_2)_{\epsilon} \# (\sigma_2)_{\epsilon})) \;\; \text{[using {8. from Fact~\ref{basicprops}}]} \\ 
& = \FGM(\rho_1, \sigma_1) \cdot \FGM(\rho_2, \sigma_2),   
\end{align} 
finishing the proof.
 
\end{proof}

\begin{lemma}[Unitary invariance] 
For all quantum states $\rho $ and $\sigma $ and any unitary $U$, we have 
\begin{equation} 
	\FGM(U \rho U^\dagger , U \sigma  U^\dagger ) = \FGM(\rho ,\sigma ).  \label{item:unitary} 
\end{equation}  
\end{lemma} 
  
\begin{proof} 
Let $W$ be an optimal solution to the SDP 
\begin{equation} 
\FGM(\rho ,\sigma ) = \max \left\{ \Tr(W) :   
\left[ \begin{array}{cc} 
\rho & W \\ 
W & \sigma  
\end{array} \right] \succcurlyeq 0 
\right\}. 
\end{equation} 
For a fixed unitary $U$, we have\footnote{This follows because $M\succcurlyeq 0 \iff XMX^{\dagger} \succcurlyeq 0$ for invertible $X$, where in this case $X = \begin{pmatrix} U & 0 \\ 0 & U \end{pmatrix}$.}
\begin{equation} 
\left[ \begin{array}{cc} 
\rho & W \\ 
W & \sigma  
\end{array} \right] \succcurlyeq 0 
\, \text{ if and only if } \, 
\left[ \begin{array}{cc} 
U \rho U^\dagger  & U W U^\dagger  \\ 
U W U^\dagger  & U \sigma  U^\dagger 
\end{array} \right] \succcurlyeq 0.  
\end{equation} 
Thus $W' = UWU^\dagger $ is feasible for the SDP 
\begin{equation} 
\FGM(U \rho U^\dagger, U \sigma U^\dagger) = \max \left\{ \Tr(W') :   
\left[ \begin{array}{cc} 
U \rho U^\dagger  & W'  \\ 
W'  & U \sigma  U^\dagger 
\end{array} \right] \succcurlyeq 0 
\right\}  
\end{equation}  
implying 
\begin{equation} \label{unitary}
	\FGM(U \rho U^\dagger , U \sigma  U^\dagger ) \geq 
	\Tr(W') 
	= \Tr(W)
	= \FGM(\rho ,\sigma ).  
\end{equation} 
For the inverse unitary $V := U^\dagger$  
we have from \eqref{unitary} that  
\begin{equation} 
	\FGM(\rho ,\sigma ) 
= 
\FGM(VU \rho U^\dagger V^\dagger , VU \sigma  U^\dagger V^\dagger ) 
\geq 
\FGM(U \rho U^\dagger , U \sigma  U^\dagger )
\geq 
\FGM(\rho ,\sigma ) 
\end{equation} 
concluding the proof. 
\end{proof}
  
\begin{lemma}[Monotonicity under PTP maps] 
For any quantum states $\rho $ and $\sigma $ and PTP (positive, trace-preserving) map $\Phi$, we have 
\begin{equation}
	\FGM(\Phi(\rho ), \Phi(\sigma )) \geq \FGM(\rho ,\sigma ). \label{item:dpi} 
\end{equation} 
\end{lemma} 

\begin{proof}  
By \Cref{handylemma}, we know $W = \Phi(\rho_{\eps}) \# \Phi(\sigma_{\eps})$ satisfies   
\begin{equation} 
\left[ \begin{array}{cc} 
\Phi(\rho _{\eps}) & W \\ 
W & \Phi(\sigma _{\eps}) 
\end{array} \right] \succcurlyeq 0    
\end{equation} 
recalling the shorthand notation~\Cref{shorthand}. 
Since $\Phi$ is a linear operator, we have 
\begin{equation}  \label{ordering}
	\Phi(\rho _{\eps}) = \Phi(\rho  + \eps \Id) = \Phi(\rho ) + \eps \Phi(\Id) \preccurlyeq \Phi(\rho ) + \eps t \Id = \Phi(\rho )_{t \eps}  
\end{equation}  
where $t = \| \Phi(\Id) \|_\infty$. 
Similarly, $\Phi(\sigma _{\eps}) \preccurlyeq \Phi(\sigma )_{t \eps}$. 
Thus, $W = \Phi(\rho_{\eps}) \# \Phi(\sigma_{\eps})$ also satisfies 
\begin{equation} 
\left[ \begin{array}{cc} 
		\Phi(\rho )_{t \eps} & W \\ 
		W & \Phi(\sigma )_{t \eps}  
\end{array} \right] \succcurlyeq 0. 
\end{equation} 
By \Cref{handylemma} again, this implies that 
\begin{equation}
\Phi(\rho)_{t \eps} \# \Phi(\sigma)_{t \eps}
\succeq 
\Phi(\rho_{\eps}) \# \Phi(\sigma_{\eps}).  
\end{equation}  
 
By a result by Ando~\cite{Ando} (and mentioned previously in \Cref{basicprops}), 
we have that 
\begin{equation} 
\Phi(\rho _{\eps}) \# \Phi(\sigma _{\eps}) \succeq \Phi(\rho _{\eps} \# \sigma _{\eps}). 
\end{equation}
Combining the above two inequalities, we have  
\begin{equation} \label{2nd}
\Tr(\Phi(\rho)_{t \eps} \# \Phi(\sigma)_{t \eps})  
\geq 
\Tr(\Phi(\rho_{\eps}) \# \Phi(\sigma_{\eps}))   
\geq 
\Tr(\Phi(\rho _{\eps} \# \sigma _{\eps})) 
= 
\Tr(\rho _{\eps} \# \sigma _{\eps})   
\end{equation} 
since $\Phi$ is trace-preserving. 
Taking limits finishes the proof. 
\end{proof} 

Note that this is a property shared with the \standard fidelity as shown in \cite{dpi}. 
We stress here that PTP maps are more general than quantum channels as completely positivity is a stronger condition than positivity.   
     
\begin{lemma}[Joint concavity] 
For any quantum states $\rho _1, \ldots, \rho _n$ and $\sigma _1, \ldots, \sigma _n$ and probability distribution $\{p_i\}$, we have 
\begin{equation} 
\FGM \left( \sum_{i=1}^n p_i \rho  _i, \sum_{i=1}^n p_i \sigma  _i \right) 
\geq \sum_{i=1}^n p_i \FGM (\rho  _i, \sigma  _i). \label{item:concavity} 
\end{equation} 
\end{lemma} 

\begin{proof}  
From \Cref{handylemma}, we have that  
\begin{equation} 
\left[ \begin{array}{cc} 
\rho _i & \rho _i \# \sigma _i \\ 
\rho _i \# \sigma _i & \sigma _i 
\end{array} \right] \succcurlyeq 0  
\end{equation} 
for all $i \in \{ 1, \ldots, n \}$. 
Using the fact that positive semidefinite matrices form a convex set, we have 
\begin{equation} 
\sum_{i=1}^n p_i \left[ \begin{array}{cc} 
\rho _i & \rho _i \# \sigma _i \\ 
\rho _i \# \sigma _i & \sigma _i 
\end{array} \right] 
= 
\left[ \begin{array}{cc} 
\sum_{i=1}^n p_i \rho _i & \sum_{i=1}^n p_i ( \rho _i \# \sigma _i ) \\ 
\sum_{i=1}^n p_i ( \rho _i \# \sigma _i ) & \sum_{i=1}^n p_i \sigma _i 
\end{array} \right] \succcurlyeq 0. 
\end{equation} 
This implies that 
$W := \sum_{i=1}^n p_i ( \rho _i \# \sigma _i )$ is feasible in the SDP 
\begin{equation}
\FGM \left( \sum_{i=1}^n p_i {\rho _i}, \sum_{i=1}^n p_i {\sigma _i} \right) 
= 
\max \left\{ \tr(W) : \left[ \begin{array}{cc} 
\sum_{i=1}^n p_i \rho _i & W \\ 
W & \sum_{i=1}^n p_i \sigma _i 
\end{array} \right] \succcurlyeq 0 \right\}.  
\end{equation}  
Therefore, 
\begin{equation}
\FGM \left( \sum_{i=1}^n p_i {\rho _i}, \sum_{i=1}^n p_i {\sigma _i} \right) \geq \Tr(W) = \sum_{i=1}^n p_i \Tr(\rho _i \# \sigma _i) = \sum_{i=1}^n p_i \FGM(\rho _i,  \sigma _i)\end{equation} 
as desired.   
\end{proof} 

\begin{lemma}[First Fuchs-van de Graaf inequality] 
For any quantum states $\rho $ and $\sigma $, we have 
\begin{equation} \label{item:fvg1}
	\FGM(\rho ,\sigma )^2 + \Delta (\rho ,\sigma )^2 \leq 1 
\end{equation} 
where $\Delta (\rho ,\sigma ) = \frac{1}{2} \|\rho  -\sigma  \|_1$ is the trace distance.  
\end{lemma} 

\begin{proof}
Since the inequality holds for the \standard fidelity, the result follows from Lemma~\ref{fidbound}. 
\end{proof}

Now we prove the claim made in the introduction that both these fidelities are quantizations of the classical fidelity. 

\begin{lemma}[Classical limit] \label{classicallimit}
	For any quantum states $\rho $ and $\sigma $ \emph{that commute}, we have 
\begin{equation} \label{item:classical} 
\F(\rho ,\sigma ) =	\FGM(\rho ,\sigma ) = \FCL(\left\{ p_i \right\},\left\{ q_i \right\})
\end{equation} 
where $\left\{ p_i \right\}$ are the eigenvalues of $\rho $ and $\left\{ q_i \right\}$ are the eigenvalues of $\sigma $. 
\end{lemma} 

\begin{proof} 
If $\rho$ and $\sigma$ commute, then we have that $\rho^{1/2} \sigma^{1/2}$ is positive semidefinite. 
Therefore, we have 
\begin{equation} \label{useful}
\F(\rho, \sigma) = \| \rho^{1/2} \sigma^{1/2} \|_1 = \Tr(\rho^{1/2} \sigma^{1/2}). 
\end{equation}
Thus, $X = \rho^{1/2} \sigma^{1/2}$ is an optimal solution to the SDP~(\ref{eq:FSDP}).
Since $X$ is also Hermitian, it is also an optimal solution to the SDP~(\ref{eq:FGSDP}), and thus $\F(\rho, \sigma) = \FGM(\rho, \sigma)$. 
Checking that they both equal $\FCL(\left\{ p_i \right\},\left\{ q_i \right\})$ follows by a simple calculation which can be seen from  \Cref{useful}.  
\end{proof} 
      
      
\subsection{Properties not shared with the \standard fidelity}

We now discuss properties satisfied by the \newfidelity but \emph{not} satisfied by the 
\standard fidelity. 
The following are well-known properties of the \standard fidelity function and we refer the reader to the book \cite{N&C} for further details. 

\begin{fact}[Orthogonality] 
For all quantum states $\rho $ and $\sigma $, the \standard fidelity satisfies  
\begin{equation} \label{item:zero} 
	\F(\rho ,\sigma ) = 0 
\; \text{ if and only if } \; 
\Tr(\rho  \sigma ) = 0,  
\end{equation} 
i.e., if and only if $\rho $ and $\sigma $ are orthogonal with respect to the Hilbert-Schmidt inner product. 
\end{fact} 

To see how this differs for the \newfidelity, consider two \emph{non-orthogonal} but distinct pure states $\ketbra{\psi}$ and $\ketbra{\phi}$. 
From above, we have that 
\begin{align}
	\braket{\psi }{\phi } \neq 0 \implies \F(\ketbra{\psi}, \ketbra{\phi}) \neq 0,
\end{align}
but it can also be checked that
\begin{align}
	\ketbra{\psi } \neq \ketbra{\phi } \implies \FGM(\ketbra{\psi}, \ketbra{\phi}) = 0,
\end{align}
which also follows from the following lemma. 
Thus we have non-orthogonal states with $0$ \newfidelity. 
However, if we have $\rho $ and $\sigma $ that satisfy $\Tr(\rho  \sigma ) = 0$, we do have $\FGM (\rho ,\sigma ) = 0$ by \Cref{fidbound} and \Cref{lem:bounds}.  
 
We now show that instead of orthogonality, there is another property which is equivalent to the \newfidelity being $0$.  

\begin{lemma}[Distinct image property] 
For any quantum states $\rho $ and $\sigma $, we have  
\begin{equation} 
	\FGM(\rho ,\sigma ) = 0
\; \text{ if and only if } \; 
\Range(\rho ) \cap \Range(\sigma ) = \left\{ 0 \right\}. 
\end{equation}  
\end{lemma} 

\begin{proof}  
Define $W(\rho ,\sigma ) := \left\{ W \succeq 0 : \left[ \begin{array}{cc} 
\rho  & W \\ 
W & \sigma   
\end{array} \right] \succcurlyeq 0 \right\}$. 
Lemma~\ref{hardlemma} says that ${W(\rho ,\sigma ) = \{ 0 \}}$ if and only if $\Range(\rho ) \cap \Range(\sigma ) = \left\{ 0 \right\}$. 
Since $\rho \# \sigma \in W(\rho, \sigma)$ (see \Cref{handylemma}) we have that $\FGM(\rho, \sigma) = 0$ if and only if $W(\rho, \sigma) = \{ 0 \}$. 
\end{proof}

A few remarks are in order. 
This is a rather mysterious feature of a similarity measure as two distinct pure states (which may be close or far in other measures) always have $0$ \newfidelity. 
As presented in \Cref{eq:puregm}, another way to present this is 
\begin{equation}
\FGM(\ketbra{\psi }, \ketbra{\phi }) = \qty|\braket{\psi }{\phi }|^\infty, 
\end{equation} 
to better compare with \Cref{eq:pure,eq:pureh}.
There may be some applications where this behaviour is desirable.
When fixing one state as pure and maximizing the \newfidelity over the other input (belonging to some set), a strong preference is shown for the other state to be mixed, rather than being the ``wrong'' pure state. 
This {behaviour is similar to the quantum relative entropy}, and could be desirable if an application requires a strict notion of pure states being equal.

\begin{fact}[Second Fuchs-van de Graaf inequality] 
For any quantum states $\rho $ and $\sigma $, the \standard fidelity satisfies 
\begin{equation} 
	\F(\rho ,\sigma ) + \Delta (\rho ,\sigma ) \geq 1. \label{item:fvg2} 
\end{equation} 
\end{fact} 
This does not hold for the \newfidelity in general. 
In fact, it fails \emph{maximally} in the sense that for any $\delta > 0$, we can construct $\rho $ and $\sigma $ such that 
\begin{equation} 
	\FGM(\rho ,\sigma ) + \Delta (\rho ,\sigma ) < \delta.  
\end{equation} 
To see this, fix a pure state $\ket{\psi}$ and define another  pure state $\ket{\phi}$ such that 
\begin{equation}
0 < \Delta (\ketbra{\psi}, \ketbra{\phi}) = \sqrt{1 - |\braket{\psi}{\phi}|^2} < \delta. 
\end{equation}   
Since $\Delta (\ketbra{\psi}, \ketbra{\phi})$ is positive, we have $\ketbra{\psi} \neq \ketbra{\phi}$ and thus $\FGM(\ketbra{\psi}, \ketbra{\phi}) = 0$. 
Combining, we have  
\begin{equation} 
\FGM(\ketbra{\psi}, \ketbra{\phi}) + \Delta (\ketbra{\psi}, \ketbra{\phi}) < \delta. 
\end{equation} 


\begin{table}[H]
	\centering
	\begin{tabular}{|l|c|c|c|c|} \hline
		\textbf{Property} & & $F=\F$ & $F=\FH$ & $F=\FGM $  \\\hline
		Symmetry & $F(\rho ,\sigma )=F(\sigma  ,\rho  )$ & \checkmark & \checkmark\cite{affinity,bhip} &\checkmark \\
		Bounds & $0\leq F(\rho ,\sigma )\leq 1$ & \checkmark & \checkmark\cite{affinity,bhip} &\checkmark \\
		\emph{Orthogonality} & $F=0 \iff \rho  \perp \sigma $ & \checkmark & \checkmark\cite{prettygood,affinity} &\color{red}{X} \\
		\textit{Distinct image} & $ F=0 \iff \Range (\rho ) \cap \Range (\sigma ) = \left\{ 0 \right\} $ & \color{red}{X} & \color{red}{X} & \checkmark \\ 
		Unity condition & $F=1 \iff \rho  = \sigma  $ & \checkmark & \checkmark\cite{affinity} &\checkmark \\
		{Additivity }& \tiny{$F(\lambda_1 \rho  _1 \oplus \lambda_2 \rho  _2, \lambda_1 \sigma  _1 \oplus \lambda_2 \sigma  _2) = \lambda_1 F(\rho  _1,\sigma  _1) + \lambda_2 F(\rho  _2,\sigma  _2)$} & \checkmark & \checkmark\cite{affinity} &\checkmark \\
		Multiplicativity & $F(\rho  _1 \otimes \rho  _2,\sigma  _1 \otimes \sigma  _2) = F(\rho  _1,\sigma  _1) F(\rho  _2,\sigma  _2)$ & \checkmark & \checkmark\cite{affinity} &\checkmark \\
		Unitary invariance & $F(\rho ,\sigma )=F(U\rho  U^{\dagger },U\sigma  U^{\dagger } )$ & \checkmark & \checkmark\cite{affinity} &\checkmark \\
		\Dpishort & $F(\Phi (\rho ) ,\Phi (\sigma ) )\geq F(\rho ,\sigma )$ & \checkmark & \checkmark\cite{petz}\tablefootnote{Note that for $\FH$, this was only shown for completely positive $\Phi $.} &\checkmark \\
		Joint concavity & $F \left(  \sum {p_i} \rho _i, \sum {p_i} \sigma  _i \right) \geq \sum {p_i} F(\rho  _i, \sigma  _i)$ & \checkmark & \checkmark\cite{affinity} &\checkmark \\
		First F-vdG & $F^2+\Delta ^2\leq 1$ & \checkmark & \checkmark\cite{prettygood,holevo} &\checkmark \\
		\emph{Second F-vdG} & $F+\Delta \geq 1$ & \checkmark & \checkmark\cite{prettygood,holevo} &\color{red}{X} \\
		Classical limit & $[\rho  ,\sigma  ]=0 \implies F = \FCL$ & \checkmark & \checkmark\cite{affinity} &\checkmark \\
		Pure states\tablefootnote{
		In \cite{affinity}, it was claimed that $\FH(\ketbra{\psi },\ketbra{\phi })=\qty|\braket{\psi }{\phi }|$, but one can verify directly that it is the square of this quantity.}
		 & $F(\ketbra{\psi },\ketbra{\phi }) = \, \cdots  $ & $ \qty|\braket{\psi }{\phi }| $ & $ \qty|\braket{\psi }{\phi }|^2$  & $\qty|\braket{\psi }{\phi }|^\infty $ \\ \hline
	\end{tabular}
	\caption{Table summarizing the key properties of the \standard fidelity $\F$, the Holevo fidelity $\FH$, and the \newfidelity $\FGM$. 
	Italics highlight the properties where the \newfidelity differs from the other two and
F-vdG is short for Fuchs-van de Graaf. 	
	}
	\label{tab:props}
\end{table}   
     

\subsection{Other interesting connections between the \newfidelity and the Uhlmann fidelity} 
     
In this subsection, we present a few other remaining characterizations of the \newfidelity and how they relate to the Uhlmann fidelity.

\begin{lemma} \label{neat} 
For any invertible quantum states $\rho $ and $\sigma $, we have 
\begin{equation}
	\FGM(\rho ,\sigma ) = \F(\rho , U \sigma U^{\dagger }) 
\end{equation} 
where $U = \rho ^{-1/2} \sigma ^{1/2} (\sigma ^{-1/2} \rho \sigma ^{-1/2})^{1/2} = \rho ^{-1/2} \left( \sigma  \# \rho \right) \sigma ^{-1/2}$ is a unitary matrix.  
\end{lemma} 

\begin{proof}
Direct calculation. 
\end{proof} 

We now make use of the duality theory of semidefinite programming to prove the following lemma. 

\begin{lemma} \label{dualSDP} 
For any quantum states $\rho $ and $\sigma $, we have 
\begin{equation} \label{eqDual}
	\FGM (\rho ,\sigma ) = \inf \left\{ \frac{1}{2} \inner{Y}{\rho } + \frac{1}{2} \inner{Z}{\sigma } : 
\left[ \begin{array}{cc} 
Y & X \\ 
X^\dagger  & Z 
\end{array} \right] \succcurlyeq 0, \, X + X^\dagger  = 2 \Id    
\right\}. 
\end{equation} 
\end{lemma} 

\begin{proof} 
Recall from \Cref{SDPlemma} that 
\begin{equation} 
	\FGM (\rho ,\sigma ) = \max \left\{ \Tr(W) : 
\left[ \begin{array}{cc} 
\rho & W \\ 
W & \sigma  
\end{array} \right] \succcurlyeq 0   
\right\}. 
\end{equation} 
The dual to the above SDP is given as the right-hand side of \Cref{eqDual}. 
Thus, to prove that \Cref{eqDual} holds, it suffices to show that the two SDPs share the same value. 
By \emph{strong duality}, the two SDPs share the same value if the dual is bounded from below and is \emph{strictly feasible}, i.e., 
there exists dual feasible $(X,Y,Z)$ such that 
\begin{equation} 
\left[ \begin{array}{cc} 
Y & X \\ 
X^\dagger  & Z 
\end{array} \right] \succ 0. 
\end{equation} 
Since the dual is clearly nonnegative and $(X,Y,Z) := (\Id, 2 \Id, 2 \Id)$ is a strictly feasible solution, the result follows.  
\end{proof}

This dual characterization can be compared to the dual characterization of the \standard fidelity (with respect to the SDP~(\ref{eq:FSDP})) given below 
\begin{equation} \label{Alberti}
	\F (\rho ,\sigma ) = \inf \left\{ \frac{1}{2} \inner{Y}{\rho } + \frac{1}{2} \inner{Z}{\sigma } : 
\left[ \begin{array}{cc} 
Y & \Id \\ 
\Id & Z 
\end{array} \right] \succcurlyeq 0 
\right\} 
\end{equation} 
as shown in \cite{watrous2013simpler}. 
The extra freedom the dual SDP~(\ref{eqDual}) has is that we do not need to choose $X = \Id$. 
In fact, the constraint $X + X^\dagger  = 2 \Id$ can be written as $X = \Id + A$ where $A$ is anti-Hermitian, that is, $A^\dagger  = - A$. 
{Therefore, just as the SDPs characterizing the two fidelities differ only by a Hermitian constraint, the duals only differ by an anti-Hermitian variable.   
We summarize this in the SDPs below for the \newfidelity by noting that if the parts in blue are removed, one recovers SDPs for the Uhlmann fidelity     
\begin{align} 
	\FGM (\rho ,\sigma ) & = \max \left\{ \frac{1}{2} \Tr(X) + \frac{1}{2} \Tr(X^{\dagger}) : 
\left[ \begin{array}{cc} 
\rho & X \\ 
X^{\dagger} & \sigma  
\end{array} \right] \succcurlyeq 0, \; \textcolor{blue}{X \text{ is Hermitian}}    
\right\} \\ 
& = \inf \left\{ \frac{1}{2} \inner{Y}{\rho } + \frac{1}{2} \inner{Z}{\sigma } : 
\left[ \begin{array}{cc} 
Y & \Id + \textcolor{blue}{A} \\ 
\Id - \textcolor{blue}{A}  & Z 
\end{array} \right] \succcurlyeq 0, \; \textcolor{blue}{A \text{ is anti-Hermitian}}    
\right\}.
\end{align} 
        

\subsection{Special cases} \label{sect:cases}
   
In this subsection, we further explore the behaviour of the \newfidelity in some special cases and examples.
First, we show that just like the \standard and Holevo fidelities, the \newfidelity takes a simplified form when one state is pure.
We then show some numerical examples of the differences in the behaviours of these three quantities for the case when both states are qubits. 
 

\subsubsection{One state is pure} \label{sect:onepure}
We have already discussed the case where both states are pure, but each fidelity discussed in this work also has a simple form when one of the states {is} pure and the other is mixed. 

\begin{lemma} \label{lem:pure}
	For a pure state $\ketbra{\psi }$ and a positive definite quantum state $\rho $, the \standard, Holevo, and \newfidelities take the following forms:
\begin{align}
	\F\left( \rho ,\ketbra{\psi } \right) &= \bra{\psi } \rho \ket{\psi }^{1/2}, \\
	\FH\left( \rho ,\ketbra{\psi } \right) &= \bra{\psi } \rho ^{1/2} \ket{\psi }, \\
	\FGM\left( \rho ,\ketbra{\psi } \right) &= \bra{\psi } \rho ^{-1} \ket{\psi }^{-1/2}.  
	\label{<+label+>}
\end{align}
\end{lemma}

\begin{proof}
The Uhlmann and Holevo fidelities can be seen from a direct calculation. 
From \Cref{techlem1}, we know that if $W$ satisfies 
\begin{equation} \label{purematrix}
\left[ \begin{array}{cc} 
\ketbra{\psi} & W \\ 
W & \rho 
\end{array} \right] \succcurlyeq 0 
\end{equation}  
then $W = \alpha \ketbra{\psi}$ for some $\alpha \in \R$.  
By taking Schur complements (see \Cref{ForSchur}), we know that $W$ (and hence $\alpha$) satisfies \Cref{purematrix} if and only if 
\begin{equation} 
\ketbra{\psi} \geq \alpha^2 \ketbra{\psi} \rho ^{-1} \ketbra{\psi}. 
\end{equation} 
This is obviously equivalent to 
\begin{equation} 
1 \geq \alpha^2 \bra{\psi} \rho ^{-1} \ket{\psi}. 
\end{equation}  
Maximizing over $\alpha$ yields the result.  
\end{proof}  
   

\subsubsection{Qubits}\label{sect:qubits}
       
A comparison of numerical behaviour of the fidelities is shown in \Cref{fig:compare} for a range of qubit states.
In the first row of diagrams, one state is fixed as pure and the other varies throughout the Bloch sphere. In the second row the first state is instead fixed with eigenvalues $\frac{3}{4}$ and $\frac{1}{4}$.    
  
\begin{figure}[H]
	\centering
	Fidelities for $\rho  =\ketbra{0}$; $\sigma  = \frac{1+\lambda}{2} \ketbra{\theta } + \frac{1-\lambda}{2}  \ketbra{\theta ^\perp}$: 
	\includegraphics[width=\linewidth]{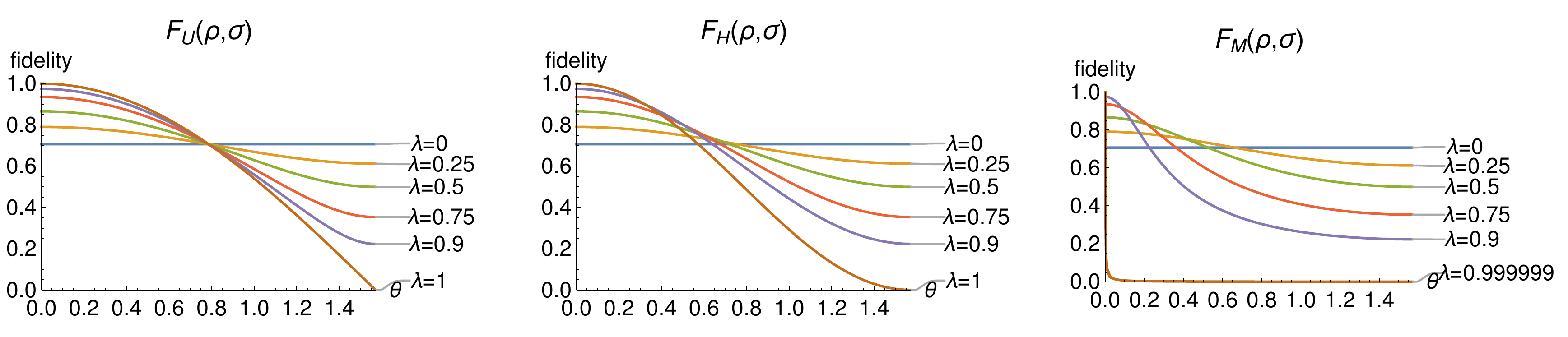}

	Fidelities for $\rho  =\frac{3}{4}\ketbra{0} + \frac{1}{4}\ketbra{1}$; $\sigma  = \frac{1+\lambda}{2} \ketbra{\theta } + \frac{1-\lambda}{2}  \ketbra{\theta ^\perp}$: 

	\includegraphics[width=\linewidth]{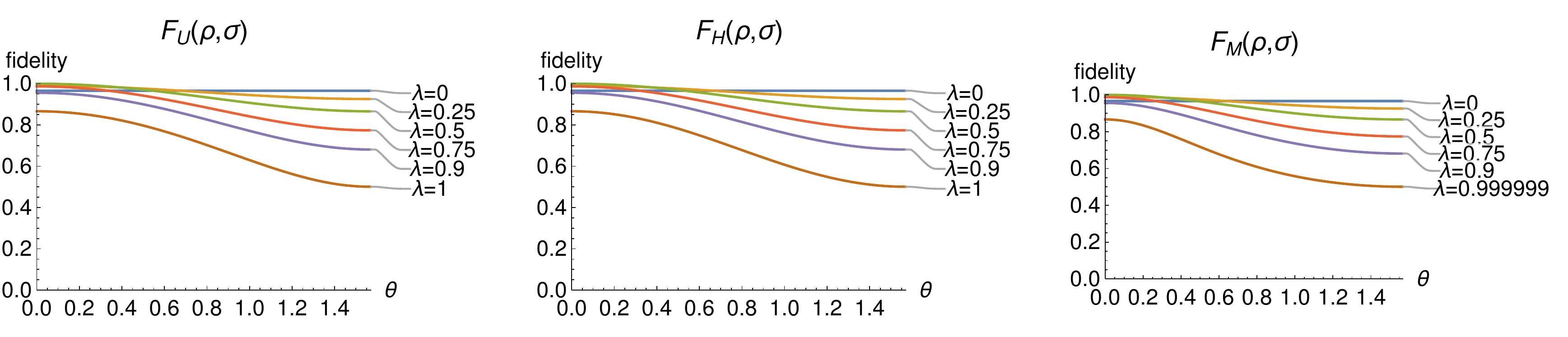}
\caption{Example plots of the three fidelities between two qubits. Here $\ket{\theta } $ is $\cos\theta \ket{0}+\sin\theta \ket{1}$, so that $\ket{\theta =0}=\ket{0}$ and $\ket{\theta =\pi /2}=\ket{1}$.
	These graphs are invariant under any global rotation in the Bloch sphere because of the unitary invariance property.
$\lambda $ represents the length of the vector in the Bloch sphere, such that $\lambda =1$ is a pure state and $\lambda =0$ is the maximally mixed state. 
We note that the fidelities are very similar when at least one state is significantly mixed (i.e.\ $\lambda \lesssim 0.5$) and the only significant discrepancies arise near the ``quantum limit'' of both states being close to pure.
}
	\label{fig:compare}
\end{figure}
   
The most obvious trend from these plots is that when either state is sufficiently mixed, the fidelities are all very similar (in line with our expectations from the previous discussions about this being the classical limit).
It is also apparent from the top-right plot that $\FGM $ only gets close to $0$ in extreme cases. Even a state with Bloch vector length $0.9$ has moderate fidelity with a pure state regardless of angle.
We explain this by more closely studying the geometry of the \newfidelity for qubits in the next section.  
         
\section{Geometric intepretation} 

The space of positive definite matrices can be pictured as a cone like the one shown in \Cref{fig:cone}.

\label{sect:geometry}
\begin{figure}[h!]
	\centering
	\def\svgwidth{\columnwidth}\vspace{-10pt}
	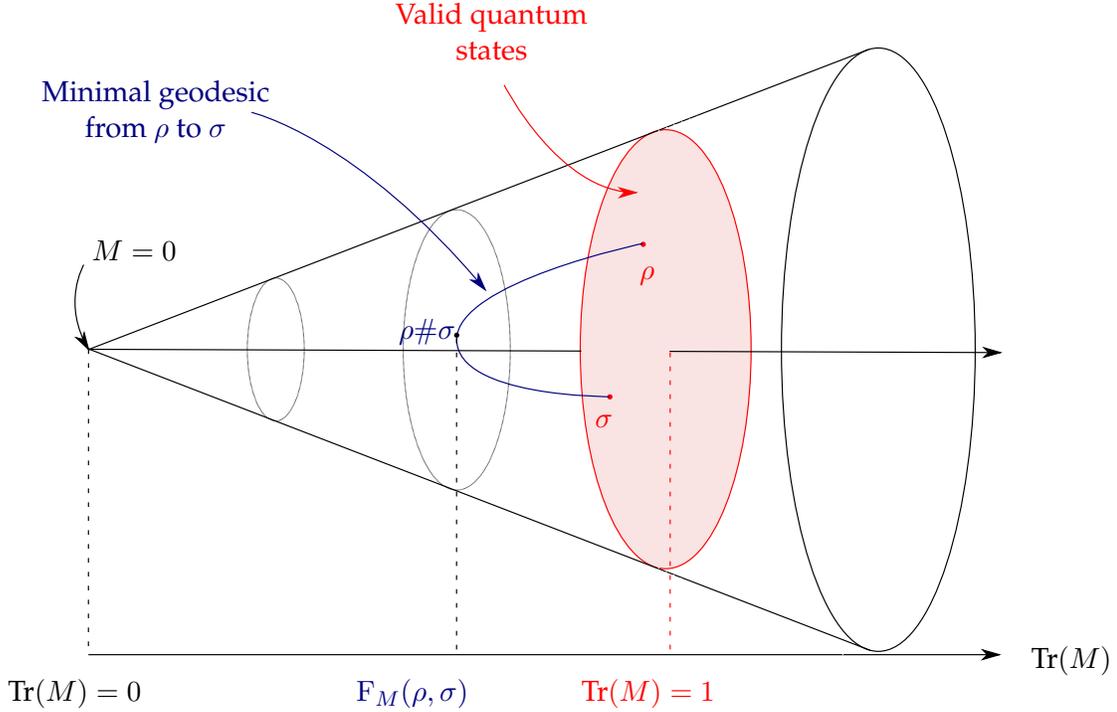 \vspace{-20pt}
	\caption{The space of positive semidefinite matrices is presented pictorially as a cone with boundary, embedded into the ambient space of Hermitian $n\times n$ matrices.
	The positive definite matrices form the interior of the cone, whereas singular matrices form the boundary (because an infinitesmal perturbation can change a zero eigenvalue to negative, putting it outside the cone).
	The central axis represents scalar multiples of the identity, such as the maximally mixed state.
The distance measure in \Cref{eq:distance} defines geodesics (i.e.\ shortest-length paths) within this conical space.
These geodesics always curve toward lower trace (leftward in the picture).
The \newfidelity, denoted by $\FGM $ -- the trace of the midpoint of this geodesic -- lies between $0$ and $1$, and measures the ``closeness'' between states according to how far leftward the geodesic curves. 
}
	\label{fig:cone}
\end{figure}
     
This space has a unique invariant Riemannian metric\footnote{Considering the exponential map from Hermitian matrices to positive definite matrices, this metric can be obtained as the push-forward of the Hilbert-Schmidt metric on the space of Hermitian matrices.} (see e.g.\ \cite{bhatia}), with the metric tensor $g$ defined {at a particular point (i.e.\ matrix) $M$} by 
\begin{align}
	g(\rho ,\sigma )|_M &= \tr(M^{-1} \rho M^{-1} \sigma ), \text{ for positive definite matrices }\rho ,\sigma ,M.
	\label{eq:metric}
\end{align}

With respect to this metric, the distance between positive definite matrices $\rho $ and $\sigma $ is 
\begin{align}
	\delta (\rho ,\sigma ) &= \left\| \log(\rho ^{-1/2} \sigma  \rho ^{-1/2})\right\|_F,
	\label{eq:distance}
\end{align}
where $\|\cdot\|_F$ is the Frobenius norm, defined as $\left\| A \right\|_F = \sqrt{\tr(A^\dagger A)}$. 
This is uniquely invariant as a distance measure in that it satisfies
\begin{align}
	\delta (X\rho X^{\dagger },X\sigma X^{\dagger }) = \delta (\rho ,\sigma ), \text{ for any invertible matrix }X.
	\label{eq:invariance}
\end{align} 
   
Using this metric, the matrix geometric mean $\rho  \# \sigma  $ is the midpoint of the minimal geodesic connecting $\rho  $ to $\sigma  $.
Equivalently, it is the unique matrix $\tau $ minimizing the least-squares distance 
\begin{equation}
\delta ^2(\rho  ,\tau )+\delta ^2(\tau ,\sigma  ). 
\end{equation} 
As shown in \Cref{fig:cone}, this geodesic curves towards the tip of the cone (the $0$ matrix), and the \newfidelity is a measure of how far it curves (i.e.\ how small the trace of the midpoint becomes). 
Quantum states that are close together (with respect to this metric) in the space of all quantum states have a geodesic which does not deviate far from that space, and so the trace of the midpoint is close to unity.     

However, it is clear from \Cref{eq:distance} that for positive semidefinite matrices that are not invertible, the metric is degenerate.
These matrices live on the boundary of the cone. 
This leads to the peculiar properties of the \newfidelity for pure states that we discussed {above}, and more generally for states with singular density matrices.

\subsection{Qubits}
The metric in \Cref{eq:metric} takes a particularly simple form for qubits, using the following parameterization for positive definite $2\times 2$ matrices.
\begin{align}
	\rho (\alpha ,r ,\theta ,\phi ) &=  e^{i \phi \sigma  _z} e^{i \theta \sigma  _y} 
	\begin{pmatrix}
		e^{- \frac{\alpha +r }{\sqrt{2}}} & 0 \\ 
		0 & e^{-\frac{\alpha -r }{\sqrt{2}}} 
	\end{pmatrix}  e^{- i \theta \sigma  _y}  e^{- i \phi \sigma  _z}
	, \text{ where } \alpha ,r ,\theta ,\phi  \in \R \\
	 &=  U D U^\dagger ,\text{ where }
	U = e^{i \phi \sigma  _z} e^{i \theta \sigma  _y}\text{ and } 
	D = \begin{pmatrix}
		e^{- \frac{\alpha +r }{\sqrt{2}}} & 0 \\ 
		0 & e^{-\frac{\alpha -r }{\sqrt{2}}} 
	\end{pmatrix}.
	\label{eq:parametrization}
\end{align} 
Then it can be shown that the metric in \Cref{eq:metric} becomes:
\begin{align}
	\mathrm{d}s^2 &= \tr\left( \rho  ^{-1} \mathrm{d}\rho  \rho ^{-1} \mathrm{d}\rho  \right) \\
	&= \mathrm{d}\alpha ^2+\mathrm{d}r ^2 + \sinh^2\!r \ \left( \mathrm{d} \theta  ^2 + \sin^2\!\theta  \ \mathrm{d}\phi ^2 \right).
\end{align}

The metric for $\left( r  ,\theta,\phi  \right)$ can be recognized as three-dimensional hyperbolic space in radial coordinates, meaning that the geometry of $2\times 2$ positive definite matrices with this metric is $\R \times \mathbb{H}_3$.

Now $r $ can be understood as parameterizing the purity of the state; with $r \to \infty $ for a pure state and $r =0$ for the maximally mixed state.
The parameter $\alpha $ is fixed for a quantum state once $r $ is determined due to the unit trace condition, as the trace of $\rho (\alpha ,r ,\theta ,\phi )$ is given by 
\begin{align}
	\tr(\rho (\alpha ,r ,\theta ,\phi )) = 2e^{- \frac{\alpha }{\sqrt{2}}}\cosh(\frac{r }{\sqrt{2}}). 
\end{align}
For quantum states, we have $\alpha = \alpha _q(r ):= -\sqrt{2} \log(\frac{1}{2} \cosh\frac{r }{\sqrt{2}}) $ so that $\rho (\alpha_q(r),r,\theta ,\phi )$ has trace $1$.
The angular coordinates $\theta $ and $\phi $ are analogous to the angular coordinates of the Bloch sphere.

Effectively, the $\sinh^2r $ prefactor in front of the angular coordinates in the metric means that a curve is always shorter if it bends ``inwards'' towards lower $r $.
A result of this is that the geodesic between two quantum states passes through states with smaller $r $ but the same $\alpha $, and hence has trace less than $1$ -- i.e.\ the \newfidelity is less than $1$ (see \Cref{fig:cone}).

\begin{figure}[h!]
	\centering
	\includegraphics[scale=1]{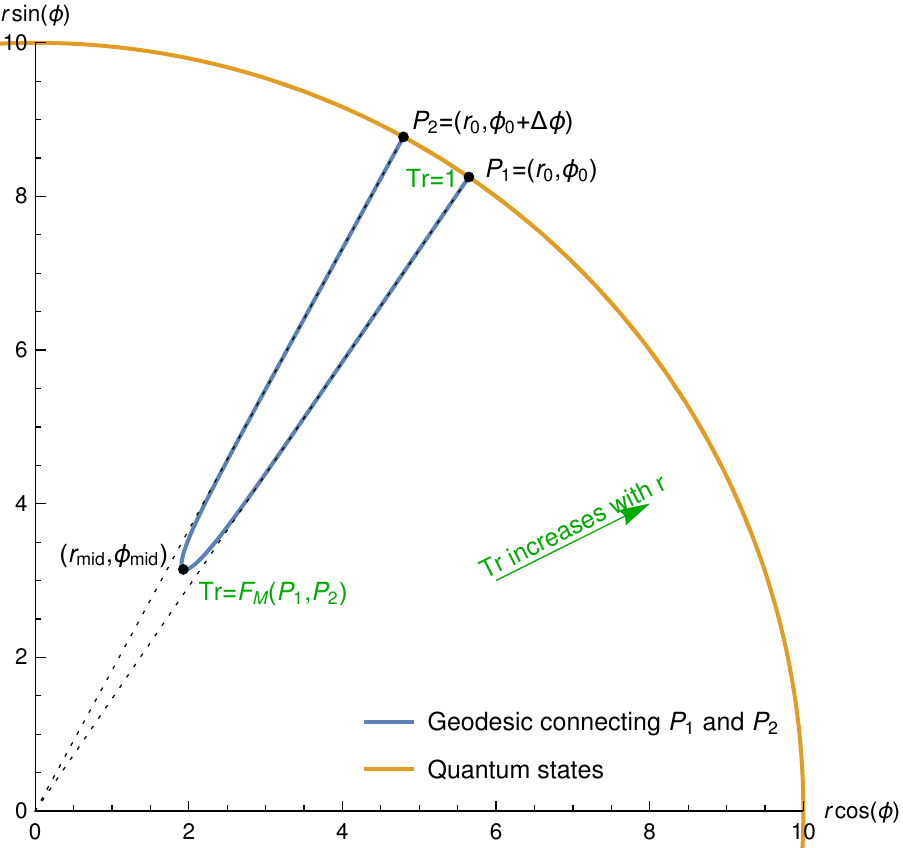}
	\caption{
	An example of the geometric interpretation of the \newfidelity.
		This is for the example discussed in the text: two qubit states with the same purity (characterized by the radial coordinate $r $) and angular coordinate $\phi $ (representing the angular coordinate separating them on the Bloch sphere). 
		However, unlike the Bloch sphere, only the orange line (at fixed $r =r _0$) represents valid quantum states with trace $1$; all other points in the plot are $2\times 2$ positive definite matrices with non-unit trace.
		The space is associated with a hyperbolic geometry, so that the minimal geodesic between the points labelled $\rho _1$ and $\rho _2$ (shown in blue) is not a straight line, but rather bends inwards.
		The trace at the midpoint along the geodesic (i.e.\ where $r =r _{\textrm{mid}}$) is the \newfidelity, and scales as $\exp(r _{\textrm{mid}} - r _0)$.
		As the states approach purity, $r_0 $ diverges to infinity, and $r _{\textrm{mid}}$ remains fixed as explained in the text, so the trace goes to zero regardless of how small $\Delta \phi $ is.
		This plot was made using $r_0 =10$, $\Delta \phi =0.1$, and the equation for the geodesic connecting the two points (with $\phi _0=0$ for convenience) is parameterized by the equation   
$		
		r (\phi ) = \arctanh\qty(\frac{\tanh(r_0)}{\cos(\phi )-\frac{\sin(\phi ) (\cos(\Delta \phi  )-1)}{\sin(\Delta \phi )}})$. 		
}
	\label{fig:hyper}
\end{figure}
   
Let us use this geometrical picture to understand why the \newfidelity of two almost-identical pure states is $0$.
Suppose we have two nearly-pure states with identical $r = r _0$ (which we eventually take to infinity so that the states become pure), and $\phi $ differing by a fixed (arbitrarily small) $\Delta \phi $.
Fix $\alpha = \alpha _q (r _0)$ and $\theta   = \frac{\pi }{2} $ for simplicity.
With these constraints, the geodesic between the states is restricted to a two-dimensional subspace parametrized by $r $ and $\phi $, with the reduced metric
\begin{align}
	\mathrm{d}s^2 = \mathrm{d}r ^2 + \sinh^2r \ \mathrm{d}\phi ^2.
	\label{<+label+>}
\end{align}
This is exactly the radial coordinates for the hyperbolic plane $\mathbb{H}_2$.

To determine the \newfidelity of these two states, we need to find the trace of the midpoint of the geodesic connecting them in this space.
Solving the geodesic equation gives the curve shown in \Cref{fig:hyper}, which curves inwards toward the centre.
One can show that the midpoint $(r _{\textrm{mid}},\phi _{\textrm{mid}})$ of the geodesic is at $r _{\textrm{mid}} = \arctanh\left(\tanh (r _0) \cos \left(\frac{\Delta \phi  }{2}\right)\right)$ and $\phi _{\textrm{mid}} = \phi _0 + \frac{1}{2} \Delta \phi $.
For large $r_0$, the former goes as
\begin{align}
	r _{\textrm{mid}}  =  \arctanh\cos \frac{\Delta \phi }{2} +O\left(e^{-2r _0}\right). 
	\label{<+label+>}
\end{align}
So for large $r _0$, $r _{\textrm{mid}}$ becomes independent of $r _0$.
This means that the minimum radius $r_{\textrm{mid}}$ of the geodesic shown in \Cref{fig:hyper} remains fixed even as $r_0 \to \infty $. 

Now let us evaluate the trace of this midpoint in order to determine the \newfidelity.
On this subspace, and at large $r_0$, the trace is
\begin{align}
	\tr (\rho (\alpha _q(r_0), r, \frac{\pi }{2}, \phi )) = \frac{\cosh( r / \sqrt{2})}{\cosh( {r_0 }/{\sqrt{2}})} 
= \frac{1}{2}e^{-\frac{r _0}{\sqrt{2}}} \cosh(\frac{r }{\sqrt{2}}) + O(e^{-\sqrt{2}r _0}).
\end{align}
Thus the \newfidelity of these two states is
\begin{align}
	\FGM\left(\rho \left(\alpha _q(r _0),r _0,\frac{\pi }{2}, \phi _0 \right),\rho \left(\alpha _q(r _0),r _0,\frac{\pi }{2}, \phi _0 + \Delta \phi  \right) \right) = f(\Delta \phi ) e^{-\frac{r_0}{\sqrt{2}}} + O(e^{-\sqrt{2}r_0}),
	\label{<+label+>}
\end{align}
with $f(\Delta \phi ) =  \frac{1}{2} \cosh( \frac{1}{\sqrt{2}} \arctanh\cos \frac{\Delta \phi }{2} )$ independent of $r_0$.
Thus for fixed $\Delta \phi $ and arbitrarily large $r _0$, we can see that this goes to $0$, demonstrating that the \newfidelity of two distinct pure states is $0$.
We can also see why it goes to $0$ so slowly when states are almost pure and almost identical, as shown in the top-right diagram of \Cref{fig:compare}; one can verify that $f(\Delta \phi )$ diverges to infinity as $\Delta \phi \to 0$, meaning that $r_0$ needs to become \textit{very} large to suppress this prefactor -- i.e.\ the states need to be ``almost pure'' before the strange behaviour of similar states having negligible fidelity occurs.
A similar argument to the one presented here can be constructed for non-qubit states. 
{This example demonstrates how the geometric picture can be useful in understanding the behaviour of the \newfidelity.}
     

\section{Conclusion and future work}\label{sect:conclusion}

In this work, we have explored the behaviour of the \newfidelity through the lens of semidefinite programming and motivated it by its connection to the geometry of positive definite matrices. 
In particular, through the semidefinite programming formulation, many proofs were simple due to the fact we do not have to worry about its limiting nature when dealing with non-invertible quantum states. 
We showed that this fidelity possesses many desirable properties one wishes to have when defining a similarity measure for quantum states. 

Since the Uhlmann fidelity function is used all over quantum theory, there is a grand landscape to see where the \newfidelity function could lend itself to be useful. 
For a concrete example, in \cite{watrous2013simpler} it was shown how to formulate the completely bounded norm of a superoperator using a characterization  
involving the \emph{maximum output fidelity}, defined as follows. 
For two positive maps $\Phi$ and $\Psi$, we define the maximum output fidelity as 
\begin{equation} 
\mathrm{F}_{\max}(\Phi, \Psi) = \max \{ \F(\Phi(\rho), \Psi(\sigma)) \} 
\end{equation}
where $\rho$ and $\sigma$ are quantum states.  
Thus, changing the fidelity above to the \newfidelity leads to a different norm-like function on superoperators. 
Considering there are fewer similarity/distance measures for quantum channels as there are for quantum states, this line of research could prove itself to be rewarding. 

\section*{Acknowledgments} 
We thank J\k{e}drzej Kaniewski and Francesco Buscemi for helpful discussions and comments on the first arxiv version and we thank William Donnelly for discussions about the geometric meaning of the \newfidelity for qubits. 
We also thank the Quantum Information and Quantum Foundations group members at the Perimeter Institute as well as many members of the Institute for Quantum Computing for fruitful discussions. 
In particular, we thank Vern Paulsen and Mizanur Rahaman for pointing out old math papers which discuss the matrix geometric mean (in particular the work \cite{Ando}).
Finally, we thank Mark M. Wilde for pointing us to the works by Matsumoto.

Research at Perimeter Institute is supported in part by the Government of Canada through the Department of Innovation, Science and Economic Development Canada and by the Province of Ontario through the Ministry of Economic Development, Job Creation and Trade. 

SC is grateful for support from the Knight-Hennessy Scholars program and the Perimeter Scholars International program.
   
\bibliography{refs}
   

\appendix 


\section{Technical lemmas about block matrices} 
\label{appendix}
    
Here we present some technical lemmas which are independent of the geometric mean but are useful for some of our proofs.

\begin{fact}[See, e.g.,\cite{bhatiasOTHERbook}] 
\label{Lemma26} 
For any positive semidefinite matrices $A$ and $B$, 
$X$ satisfies 
\begin{equation} 
\left[ \begin{array}{cc} 
A & X \\ 
X^\dagger  & B 
\end{array} \right] \succcurlyeq 0 
\end{equation}  
if and only if $X = A^{\half} V B^{\half}$ for some $V$ (not necessarily Hermitian) satisfying $\| V \|_{\infty} \leq 1$ (where $\|\cdot \|_\infty $ is the $\infty $-norm, i.e.\ the largest singular value of the matrix).  
\end{fact} 

The lemma above can be used to prove the following lemma. 
    
\begin{lemma} \label{techlem1}
For any positive semidefinite matrices $A$ and $B$, if $W$ satisfies 
\begin{equation} \label{equation212}
\left[ \begin{array}{cc} 
A & W \\ 
W & B 
\end{array} \right] \succcurlyeq 0,  
\end{equation}  
then $\Range(W) \subseteq \Range(A) \cap \Range(B)$. 
\end{lemma}  
 
\begin{proof}  
Suppose $W$ satisfies \Cref{equation212}. 
Then by Lemma~\ref{Lemma26}, we have that 
\begin{equation} 
W = A^{\half} V B^{\half} = B^{\half} V^\dagger  A^{\half} 
\end{equation} 
since it is Hermitian. 
Thus, the image of $W$ is contained in the image of both $A^{1/2}$ and $B^{1/2}$. 
Since the image of $A^{1/2}$  is equal to the image of $A$
and the image of $B^{1/2}$ is equal to the image of $B$, 
the result follows.   
\end{proof}

The following lemma characterizes a sufficient condition for the positive semidefiniteness of certain block matrices. 

\begin{lemma} \label{techlem2}
For any positive semidefinite matrices $A$, $B$, and $W$,  if we have $A \succcurlyeq W$ and $B \succcurlyeq W$, then we have 
\begin{equation} 
\left[ \begin{array}{cc} 
A & W \\ 
W & B 
\end{array} \right] \succcurlyeq 0. 
\end{equation}  
\end{lemma} 

\begin{proof}  
We have 
\begin{equation} 
\left[ \begin{array}{cc} 
A & W \\ 
W & B 
\end{array} \right] 
\succcurlyeq  
\left[ \begin{array}{cc} 
W & W \\ 
W & W 
\end{array} \right] 
= 
\left[ \begin{array}{cc} 
1 & 1 \\ 
1 & 1  
\end{array} \right] \otimes W 
\succcurlyeq 0.
\end{equation}  
\end{proof} 
 
The following technical lemma helps characterize when there are non-trivial feasible solutions to the SDP given in~\Cref{eq:FGMSDP}.  
        
\begin{lemma} \label{hardlemma}
For any positive semidefinite matrices $A$ and $B$, we have 
\begin{equation} 
\left\{ W \succeq 0: \left[ \begin{array}{cc} 
A & W \\ 
W & B 
\end{array} \right] \succcurlyeq 0 \right\} = \{ 0 \} 
\; \text{ if and only if } \; 
\Range(A) \cap \Range(B) = \{ 0 \}. 
\end{equation}
\end{lemma} 

\begin{proof} 
Define $W(A,B) := \left\{ W \succeq 0 : \left[ \begin{array}{cc} 
A & W \\ 
W & B 
\end{array} \right] \succcurlyeq 0 \right\}$ for brevity. 
If $\Range(A) \cap \Range(B) = \{ 0 \}$, then from Lemma~\ref{techlem1}, we have that $W(A,B) = \{ 0 \}$. 
Conversely, suppose there exists a nonzero vector $x \in 
\Range(A) \cap \Range(B)$. 
Then we see that there exists $\lambda > 0$, possibly very small, such that ${A\succcurlyeq\lambda xx^\dagger }$ and ${B \succcurlyeq \lambda xx^\dagger }$. 
Thus, by Lemma~\ref{techlem2}, we have that $\lambda xx^\dagger  \in W(A,B)$ and thus $W(A,B)$ contains a nonzero matrix. 
\end{proof} 

The following well-known fact gives a necessary and sufficient condition for the positive semidefiniteness of block matrices. 
        
\begin{fact} \label{ForSchur}
For any positive \emph{definite} matrix $B$, we have  
\begin{equation} 
\left[ \begin{array}{cc} 
A & X \\ 
X^\dagger  & B
\end{array} \right] \succcurlyeq 0 
\; \text{ if and only if } \; 
A \succeq XB^{-1}X^\dagger .  
\end{equation}
\end{fact}

We now prove \Cref{handylemma} as used in the main text, restated as \Cref{handylemmaproof}.
\begin{lemma} \label{handylemmaproof}
Given positive semidefinite matrices $P$ and $Q$, if $W$ satisfies 
\begin{equation} \label{eqlem32proof}
\left[ \begin{array}{cc} 
P & W \\ 
W & Q  
\end{array} \right] \succcurlyeq 0 
\end{equation} 
then $W \preceq P \# Q$. 
Moreover, $P \# Q$ satisfies \Cref{eqlem32proof}. 
\end{lemma} 

\begin{proof}  

Note that 
\begin{equation} \label{constraint2}
\left[ \begin{array}{cc} 
P & W \\ 
W & Q  
\end{array} \right] \succcurlyeq 0 
\iff 
\left[ \begin{array}{cc} 
P_{\eps} & W \\ 
W & Q_{\eps}  
\end{array} \right] \succcurlyeq 0, \; \forall \eps > 0  
\end{equation} 
recalling the notation from \Cref{shorthand} where $P_{\eps} := P + \eps \Id$ and $Q_{\eps}:= Q + \eps \Id$ for brevity. 
Now we can use the fact that $P_\eps$ is invertible, even if $P$ is not (i.e.\ in case $P$ is positive semidefinite but not positive definite).
By using Schur complements (see Fact~\ref{ForSchur}), we have
\begin{align} 
W \text{ satisfies } \eqref{eqlem32} 
& \iff Q_{\eps} \succcurlyeq W P_{\eps}^{-1} W, \; \forall \eps > 0 \\ 
& \iff P_{\eps}^{-1/2} Q_{\eps} P_{\eps}^{-1/2} \succcurlyeq 
(P_{\eps}^{-1/2} W P_{\eps}^{-1/2})^2, \; \forall \eps > 0 \\ 
& \implies (P_{\eps}^{-1/2} Q_{\eps} P_{\eps}^{-1/2})^{1/2} \succcurlyeq P_{\eps}^{-1/2} W P_{\eps}^{-1/2}, \; \forall \eps > 0 \\ 
& \iff P_{\eps} \# Q_{\eps} \succcurlyeq W, \; \forall \eps > 0. 
\end{align} 
Note that taking square roots preserves the partial ordering of positive semidefinite matrices, but squaring does not, and thus the third line above does not imply the second line. 
Since $W \preceq P_{\eps} \# Q_{\eps}$ for all $\eps > 0$, we have 
\begin{equation}
W \preceq \lim_{\epsilon \to 0} ( P_{\eps} \# Q_{\eps} ) = P \# Q  
\end{equation} 
since the set of positive semidefinite matrices is a closed set. 

We now show that $W = P \# Q$ satisfies \Cref{eqlem32}. 
To see this, we define the following unitary 
\begin{equation} 
U_{\eps} := Q_{\eps}^{-1/2} P_{\eps}^{1/2} (P_{\eps}^{-1/2} Q_{\eps} P_{\eps}^{-1/2} )^{1/2}. 
\end{equation} 
It is easy to check that this is indeed a unitary matrix. 
Notice also that we have 
\begin{equation} 
P_{\eps} \# Q_{\eps} = Q_{\eps}^{\half} U_{\eps} P_{\eps}^{\half} = P_{\eps}^{\half} U_{\eps}^\dagger Q_{\eps}^{\half}. 
\end{equation} 
Therefore, we have 
\begin{equation} 
\left[ \begin{array}{cc} 
P_{\eps} & P_{\eps} \# Q_{\eps} \\ 
P_{\eps} \# Q_{\eps} & Q_{\eps} 
\end{array} \right] 
= 
\left[ \begin{array}{c} 
P_{\eps}^{\half} \\ 
Q_{\eps}^{\half} U_{\eps}
\end{array} \right] 
\left[ \begin{array}{c} 
P_{\eps}^{\half} \\ 
Q_{\eps}^{\half} U_{\eps} 
\end{array} \right]^\dagger  \succcurlyeq 0.   
\end{equation} 
Again, since the set of positive semidefinite matrices is a closed set, we have that 
\begin{equation} 
\lim_{\eps \to 0} 
\left[ \begin{array}{cc} 
P_{\eps} & P_{\eps} \# Q_{\eps} \\ 
P_{\eps} \# P_{\eps} & Q_{\eps} 
\end{array} \right] = 
\left[ \begin{array}{cc} 
P & P \# Q \\ 
P \# Q  & Q  
\end{array} \right]  \succcurlyeq 0  
\end{equation} 
concluding the proof. 
 
\end{proof}

\end{document}